\pdfoutput=1
\documentclass[acmsmall,nonacm]{acmart}
\PassOptionsToPackage{dvipsnames}{xcolor}
\PassOptionsToPackage{capitalize,nameinlink}{cleveref}

\def\BibTeX{{\rm B\kern-.05em{\sc i\kern-.025em b}\kern-.08emT\kern-.1667em\lower.7ex\hbox{E}\kern-.125emX}}

\usepackage{nicefrac}
\usepackage{siunitx}
\usepackage{array,framed}
\usepackage{booktabs}
\usepackage{
  color,
  float,
  epsfig,
  wrapfig,
  graphics,
  graphicx,
  subcaption
}
\usepackage[dvipsnames]{xcolor}
\usepackage{textcomp}
\usepackage{setspace}
\usepackage{amsmath}
\usepackage{amsfonts}

\usepackage{latexsym,fancyhdr,url}
\usepackage{enumerate}
\usepackage{graphics}
\usepackage{xparse} 
\usepackage{makecell}
\usepackage{multirow}
\usepackage{csvsimple}
\usepackage{balance}
\usepackage{soul}
\usepackage{framed}
\usepackage{algorithm}
\usepackage[noend]{algpseudocode}
\usepackage{soul}
\usepackage{amsthm}
\usepackage{xspace}
\usepackage{flushend}
\usepackage{multicol}
\usepackage{mdframed}
\usepackage{longfbox}
\usepackage{enumitem}
\usepackage{microtype}
\makeatletter
\newdimen\@tempdimd
\makeatother
\fboxset{rounded,outer-width=\linewidth, border-color=red, padding-top=3pt, padding-bottom=3pt}

\usepackage{
  tikz,
  pgfplots,
  pgfplotstable
}
\usepackage{hyperref}

\algdef{SE}[EVENT]{Event}{EndEvent}[1]{\textbf{upon}\ #1\ \algorithmicdo}

\usetikzlibrary{
  shapes.geometric,
  arrows,
  external,
  pgfplots.groupplots,
  matrix
}

\pgfplotsset{compat=1.9}

\usepackage{mathtools,}

\DeclareMathAlphabet{\mathcal}{OMS}{cmsy}{m}{n}

\DeclareGraphicsExtensions{%
    .png,.PNG,%
    .pdf,.PDF,%
    .jpg,.mps,.jpeg,.jbig2,.jb2,.JPG,.JPEG,.JBIG2,.JB2}

\usepackage{cleveref}
\usepackage{xparse}
\newcommand{\bnm}{\begin{newmath}}
\newcommand{\enm}{\end{newmath}}

\newcommand{\bea}{\begin{eqnarray*}}%
\newcommand{\eea}{\end{eqnarray*}}%

\newcommand{\bne}{\begin{newequation}}
\newcommand{\ene}{\end{newequation}}

\newcommand{\bal}{\begin{newalign}}
\newcommand{\eal}{\end{newalign}}

\newenvironment{newalign}{\begin{align}%
\setlength{\abovedisplayskip}{1pt}%
\setlength{\belowdisplayskip}{1pt}%
\setlength{\abovecaptionskip}{1pt}%
\setlength{\belowcaptionskip}{1pt}%
\setlength{\abovedisplayshortskip}{1pt}%
\setlength{\belowdisplayshortskip}{1pt} }{\end{align}}
\newenvironment{newmath}{\begin{displaymath}%
\setlength{\abovedisplayskip}{4pt}%
\setlength{\belowdisplayskip}{4pt}%
\setlength{\abovedisplayshortskip}{6pt}%
\setlength{\belowdisplayshortskip}{6pt} }{\end{displaymath}}

\newenvironment{newequation}{\begin{equation}%
\setlength{\abovedisplayskip}{4pt}%
\setlength{\belowdisplayskip}{4pt}%
\setlength{\abovedisplayshortskip}{6pt}%
\setlength{\belowdisplayshortskip}{6pt} }{\end{equation}}

\newcounter{ctr}

\newcounter{mytable}
\def\mytable{\begin{centering}\refstepcounter{mytable}}
\def\endmytable{\end{centering}}

\newcounter{myfig}
\def\myfig{\begin{centering}\refstepcounter{myfig}}
\def\endmyfig{\end{centering}}

\newlength{\saveparindent}
\newlength{\saveparskip}
\newcommand{\E}{{\rm I\kern-.3em E}}

\renewcommand{\eqref}[1]{\mbox{Equation~(\ref{#1})}}

\def \part {part}

\renewcommand{\paragraph}[1]{\vspace*{6pt}\noindent\textbf{#1}\;}

\newtheoremstyle{compressed}
  {5pt} 
  {5pt} 
  {\itshape} 
  {.7em} 
  {\scshape} 
  {.} 
  {.5em} 
  {} 
\theoremstyle{compressed}
\newtheorem{theorem}{Theorem}[section]

\def \blackslug{\hbox{\hskip 1pt \vrule width 4pt height 8pt
    depth 1.5pt \hskip 1pt}}
\def \qed{\quad\blackslug\lower 8.5pt\null\par}

\newcounter{mynote}[section]

\newcommand\ignore[1]{}

\newcounter{rcnote}[section]

\newcounter{mrnote}[section]

\newcounter{fknote}[section]

\newcounter{anote}[section]

\DeclareMathSymbol{\mlq}{\mathord}{operators}{``}
\DeclareMathSymbol{\mrq}{\mathord}{operators}{`'}

\newcommand{\rhf}[2]{R_{f, \gamma}}

\DeclareDocumentCommand{\edist}{o o}{
  \ensuremath{
    \IfNoValueTF{#1}{{d}}{{\sf d}(#1,#2)}
  }
}

\newcommand{\olrk}[1]{\ifx\nursymbol#1\else\!\!\mskip4.5mu plus 0.5mu\left(\mskip0.5mu plus0.5mu #1\mskip1.5mu plus0.5mu \right)\fi}

\NewDocumentCommand{\indseq}{ O{1} O{r} }{{#1}\ldots {#2}}

\newboolean{comments}
\setboolean{comments}{false}
\ifthenelse{\boolean{comments}}{
\newcommand{\neil}[1]{\textcolor{red}{\textbf{Neil:} #1}}
\newcommand{\heidi}[1]{\textcolor{blue}{\textbf{Heidi:} #1}}
\newcommand{\natacha}[1]{\textcolor{cyan}{\textbf{Natacha:} #1}}
\newcommand{\nc}[1]{\textcolor{cyan}{\textbf{Natacha:} #1}}
\newcommand{\alin}[1]{\textcolor{ForestGreen}{\textbf{Alin:} #1}}
\newcommand{\ittai}[1]{\textcolor{brown}{\textbf{Ittai comment:} #1}}
\newcommand{\IA}[1]{\textcolor{brown}{\textbf{IA:} #1}}
\newcommand{\fs}[1]{\textcolor{teal}{\textbf{FS:} #1}}
\newcommand{\changebars}[2]{\textcolor{red}{(#1)\st{(#2)}}}
}{
\newcommand{\neil}[1]{}
\newcommand{\heidi}[1]{}
\newcommand{\natacha}[1]{}
\newcommand{\nc}[1]{}
\newcommand{\alin}[1]{}
\newcommand{\ittai}[1]{}
\newcommand{\IA}[1]{}
\newcommand{\fs}[1]{}
\newcommand{\changebars}[2]{}
}
\algnewcommand{\LineComment}[1]{\State \(\triangleright\) \textcolor{gray}{#1}}

\makeatletter
\algblockdefx[Upon]{Upon}{EndUpon}
  [1]{\textbf{upon} #1 \textbf{do}}
  {end}
\algtext*{EndUpon}

\algblockdefx[Upon]{Upon}{EndUpon}
  [1]{\textbf{upon} #1 \textbf{do}}
  {end}
\algtext*{EndUpon}

\newcommand{\msgfont}[1]{\textsc{#1}}
\newcommand{\msg}[2]{$\langle$\msgfont{#1}, $#2\rangle$}
\newcommand{\emsg}[2]{$\langle$\msgfont{#1}$\rangle$}
\algnewcommand{\algorithmicstate}{\textbf{State:}}
\algnewcommand{\GlobalState}{\item[\algorithmicstate]}

\theoremstyle{definition} 

\newcommand{\sys}{BeeGees}
\newcommand{\consec}{CHL}
\newcommand{\nonconsec}{AHL}
\newcommand{\phaseone}{Prepare}

\begin{document}
\def\thetitle{\sys{}: stayin' alive in chained BFT}
\title{\thetitle}






\author{Ittai Abraham}
\affiliation{
\institution{VMware Research}
\country{Israel}
}
\email{iabraham@vmware.com}

\author{Natacha Crooks}
\affiliation{
\institution{UC Berkeley}
\country{USA}
}
\email{ncrooks@berkeley.edu}

\author{Neil Giridharan}\authornote{Lead author.}
\affiliation{
\institution{UC Berkeley}
\country{USA}
}
\email{giridhn@berkeley.edu}

\author{Heidi Howard}
\affiliation{
\institution{Microsoft Research}
\country{UK}
}
\email{heidi.howard@microsoft.com}

\author{Florian Suri-Payer}
\affiliation{
\institution{Cornell University}
\country{USA}
}
\email{fsp@cs.cornell.edu}
%
%


\begin{abstract}
Modern \textit{chained} Byzantine Fault Tolerant (BFT) systems leverage a combination of pipelining and leader rotation to obtain both efficiency and fairness. 
%
These protocols, however, require a sequence of three or four \textit{consecutive} honest leaders to commit operations. Therefore, even simple leader failures such as crashes can weaken liveness both theoretically and practically.
%
Obtaining a chained BFT protocol that reaches decisions even if the sequence of honest leaders is non-consecutive, remains an open question.

To resolve this question we present \textit{\sys{}}\footnote{\sys{} \textit{stays (a-)live} against the odds.}, a novel chained BFT protocol that successfully commits blocks even with non-consecutive honest leaders. It does this while also maintaining quadratic word complexity with threshold signatures, linear word complexity with SNARKs, and responsiveness between consecutive honest leaders.
%
\sys{} reduces the expected commit latency of HotStuff by a factor of three under failures, and the worst-case latency by a factor of seven.
\end{abstract}

\maketitle

%



\section{Introduction}

Blockchain systems have emerged as a promising way for mutually distrustful parties to compute over shared data. Byzantine Fault Tolerant (BFT) state machine replication (SMR), the core protocol in most blockchains, provides to applications the abstraction of a centralized, trusted, and always available server. BFT SMR guarantees that a set of replicas will agree on a common sequence of operations, even though some nodes may misbehave. Blockchain systems add two additional constraints over prior work.
1) operation ordering should be \textit{fair}: it must closely follow the order in which operations are submitted, 
and offer no single party undue influence in the process. Protocols without fairness can be abused by the application: participants may censor or front-run to gain economic advantages. 2) protocols should scale to large number of replicas (in the hundreds). 

To address these concerns, recent BFT SMR protocols targeted at blockchains, such as HotStuff~\cite{yin2019hotstuff}, DiemBFTv4~\cite{diem2021}, Fast-Hotstuff~\cite{jalalzai2020fasthotstuff} as well as the largest Proof-of-Stake system, Ethereum (Casper FFG~\cite{buterin2017casper}), are structured around two key building blocks: 
    \ul{\textit{Chaining.}} Every BFT protocol requires a (worst-case) minimum of two voting rounds (henceforth \textit{phases}). Each voting phase aims to establish a \textit{quorum certificate (QC)} by collecting a set of signed votes from a majority of honest replicas. Blockchain systems \textit{pipeline} the voting phases of consecutive proposals to avoid redundant coordination and cryptography: the system can use the quorum certificate of the second voting phase of block $i$ to certify the first phase of block $i+1$.
    Each block then requires (on average) generating and verifying the signatures of a single QC. This is especially important for large participant sets as QC sizes grow linearly with the number of replicas, increasing cryptographic costs.
    
    \hspace{3pt}  \ul{\textit{Leader-Speaks-Once (LSO).}} To minimize fairness concerns associated with leader-based solutions and to decrease the influence of adaptive adversaries (who control the network), BFT protocols targeted at blockchains adopt a \textit{leader-speak-once} (LSO) model. In LSO, each leader proposes and certifies a single block after which the leader is immediately rotated out as part of a new \textit{view}. Electing a different leader per block limits the leader's influence; it can manipulate transactions in the proposed block only. Traditional BFT protocols (such as PBFT~\cite{castro1999pbft}), in contrast, adopt a \textit{stable-leader} paradigm in which leaders are only replaced if they fail to make progress through a  fallback \textit{view change} protocol. Failures are assumed to be infrequent, and thus protocol complexities (and costs) are intentionally moved into this view change, allowing for a simpler and more efficient failure-free steady case.

While a joint approach that is both \textit{chained} and \textit{leader-speak-once} (CLSO) is desirable, the combination of these two properties also introduces a new challenge: how to preserve safety when block commitment is spread across leaders? This work observes that all prior CLSO protocols solve this challenge by \textit{unintentionally} relinquishing liveness -- and proposes a novel protocol that manages to avoid this trade-off (without sacrificing performance). 

\par \textbf{The problem.} To maintain safety, block commitment in prior CLSO protocols requires a sequence of $k$ QCs in consecutive views (where $k\in \{2,3\}$ depending on protocol). Consequently, liveness is only guaranteed in the presence of $k+1$ \textit{consecutive} honest (non-faulty) leaders. In the remainder of this paper, we refer to this property as \consec{} (consecutive honest leaders).


\begin{definition} (\consec).
After GST, if an honest leader in view $v$ proposes a value, then it is guaranteed to commit if views $v+1,v+2,...,v+k$ (contiguous views) have honest leaders .
\end{definition}
This constraint introduces significant performance penalties in practice.
We show in Section~\ref{sec:background} that in for HotStuff~\cite{yin2019hotstuff} example -- the protocol at the core of the former Diem blockchain~\cite{diem} (now Aptos~\cite{aptos})  -- a single faulty leader may suffice to prevent \textit{any} block from being committed for some system configurations. Further, we show that even for arbitrary configurations, faulty leaders can always greatly reduce protocol throughput. Worse, this attack does not require any explicit equivocation; it suffices for a faulty leader to simply delay responding, making it hard to detect misbehavior --, and thus represents a significant exploit opportunity for a Byzantine attacker. To the best of our knowledge, this liveness concern is present in \textit{all} existing CLSO protocols today. This paper asks: is this fundamental or can we do better? 

\par \textbf{Our solution.} We find that yes, it is possible to improve the liveness guarantee offered by CLSO protocols. To this effect, we propose \sys{}, a new consensus protocol that strengthens liveness and instead satisfies the following stronger property we call \nonconsec{} (any honest leader):





\begin{definition} (\nonconsec).
After GST, if an honest leader in view $v$ proposes a value, then it is guaranteed to commit if views $v_{i_1}<v_{i_2}<...v_{i_k}$ (non-contiguous views) have honest leaders.
\end{definition}

When QCs are not contiguous, it becomes possible for conflicting QCs to form unbeknownst to the current leader; these QCs can trigger safety violations when committing a block.
\sys{}'s core insight lies in observing that \textit{\phaseone{}} messages (called \textsc{prepare}  messages in Hotstuff, \textsc{Pre-prepare} messages in PBFT), which are traditionally discarded by BFT protocols, can in fact be leveraged to strengthen liveness (\nonconsec{}). In the presence of omission faults or asynchrony, \sys{} uses these messages to \textit{prevent} conflicting QCs from forming. In the presence of equivocation, \sys{} instead uses \textit{\phaseone{}} messages to reliably detect when a conflicting QC could have
formed and eagerly abort block commitment.
\textit{\phaseone{}} messages further allow \sys{} to detect \textit{implicit QCs}, QCs for blocks proposed by honest leaders that would have formed but for a malicious leader failing to disseminate them. Together, these properties allow \sys{} to be the first CLSO protocol to satisfy \nonconsec{}.
Satisfying this stronger liveness property drastically curbs the impact of a Byzantine leader on the system: after GST, no node can delay the commitment of an honest leader's proposal by more than one view.
Importantly, \sys{} achieves this without sacrificing performance: it has optimal latency of two phases~\cite{good-case}, quadratic word complexity when used with threshold signatures~\cite{gueta2019sbft}, and linear word complexity with SNARKs~\cite{abspoel2020malicious}, matching the state of the art (\cite{jolteon-ditto,jalalzai2020fasthotstuff}).


\par \textbf{Understanding the limits of existing protocols.} 
\sys{} is the first protocol to satisfy the stronger liveness property \nonconsec{} while also maintaining safety. This is a result of it being the only protocol to guarantee a property called \textit{sequentiality}. In this work, we observe that sequentiality is a necessary property for \nonconsec{}, yet it is not offered by any prior CLSO protocol. This explains why prior attempts were unable to safely offer \nonconsec{}~\cite{bano2020twins}. Roughly speaking, sequentiality mandates that, after GST, for any pair of two honest leaders in views $v$ and $v'>v$, $v'$'s leader must extend the block proposed by $v$'s leader.
\begin{theorem}\label{thm:clso-seq}
\nonconsec{} CLSO is achievable only if sequentiality is satisfied.
\end{theorem}

This paper is structured as follows. We first introduce relevant background (\S\ref{sec:prelim}, \S\ref{sec:background}) before presenting \sys{}  (\S\ref{sec:protocol}).  
We experimentally validate our claims (\S\ref{sec:complexity}),
before proving that sequentiality is necessary (\S\ref{sec:sequentiality}) and concluding (\S\ref{sec:conclusion}). 
An earlier iteration of this result was published as a brief announcement~\cite{siestaBA}. Our previous effort included an algorithm that makes progress with non-consecutive leaders in some but not all settings and thus does not satisfy \nonconsec{}.
\section{Preliminaries}
\label{sec:prelim}

We adopt the standard BFT system model in which $n=3f+1$ replicas communicate through a reliable, authenticated, point-to-point network where at most $f$ replicas are faulty. A strong but static adversary can coordinate faulty replicas' actions but cannot break standard cryptographic primitives. We adopt the partially synchronous model, where there exists a known 
upper bound
$\Delta$ on the communication delay, and an unknown Global Stabilization Time (GST) after  all messages will arrive within $\Delta$ \cite{dwork1988consensus}.
We assume the availability of standard digital signatures and a public-key infrastructure (PKI). We use \emsg{m}{}$_r$ to
denote a message $m$ signed by replica $r$. A message is well-formed if all of its signatures
are valid. 


Byzantine fault-tolerant state machine replication (BFT SMR) is formally defined as follows:
\begin{definition}
(BFT SMR). A Byzantine fault tolerant state machine
replication protocol commits client requests in a linearizable log, which satisfies the following 
properties~\cite{fab-good-case}~\cite{good-case}. 
\begin{itemize}
    \item Safety. Honest replicas commit the same  values at the same log position.
    \item Liveness. All client requests eventually receive a response; all requests are eventually committed by every honest replica.
\end{itemize}
\end{definition} 

We also formalize the notion of CLSO (inspired from~\cite{rotating-leader}).
\begin{definition} (CLSO).
A CLSO protocol is a BFT-SMR protocol that proceeds in a sequence of views and has three properties: 
\begin{itemize}
\item Each view changes the leader.
\item There is an infinite number of views led by honest leaders.
\item Block commitment cannot be guaranteed within a single view.


\end{itemize}


\end{definition}

\section{Related Work \& Liveness issues in CLSO protocols}
\label{sec:background}
All existing CLSO protocols follow the same general pattern. While we focus on HotStuff here~\cite{yin2019hotstuff}, our observations broadly apply to all current CSLO protocols.

Most such protocols follow a similar logical structure~\cite{jalalzai2020fasthotstuff,diem2021,buterin2017casper,buchman2016tendermint,yin2019hotstuff,jolteon-ditto}.  They proceed in a sequence of \textit{views}. 
In each view, a designated leader proposes a batch of client operations (a \textit{block}), and drives agreement to safely order and commit these operations.
Blocks contain a parent pointer to their predecessor block, thus forming a chain. At a high level, the protocol proceeds as follows. 

\par \textbf{Normal Case.} The leader of view $v$ begins by proposing a block $B$ for log slot $i$.
Committing a proposal consists of two logical phases:  
a \textit{non-equivocation} phase and a \textit{durability} phase. The non-equivocation phase ensures that at most one block proposal will be agreed on in per view. The durability phase ensures that any (possibly) agreed-upon decision is preserved across views and leader changes, thus guaranteeing that only one block can be committed for log slot $i$. Each phase makes use of \textit{quorum certificates} (QC) to achieved the desired invariants. A QC, written $QC=(B,v,\sigma)$, refers to a set of unique signed replica votes $\sigma$ for block $B$ proposed in view $v$. A QC describes a \textit{threshold} $|\sigma|$ of confirmations proving that a super-majority of distinct replicas voted for block $B$. 
Upon committing a block, honest replicas execute its transactions and enter the next view through a view change (described below).

\par \textbf{View Change.}
The \textit{view change} is responsible for changing leaders and preserving all decisions made durable by previous leaders.
View change protocols are notoriously tricky: they can be expensive and hard to get right~\cite{bano2020twins}. The primary challenge stems from reconciling different participants' beliefs about what \textit{could have been committed}, as asynchrony and malicious leaders may cause replicas to consider different sets of blocks as potentially committed.

To understand the liveness challenges associated with CLSO protocols, we first describe in more detail how stable non-chained (\textit{basic}) Hotstuff works (\S\ref{s:hs-basic}), before introducing the refinements of \textit{chaining}/\textit{pipelining} and \textit{leader-speaks-once} (\S\ref{s:hs-cslo}). We then demonstrate the resulting liveness pitfall (\S\ref{s:live-kaput}), and show that it is non-trivial to address (\S\ref{s:relax-take-it-eeeaaasy}). 

\subsection{Basic HotStuff}
\label{s:hs-basic}
HotStuff proceeds in a sequence of \textit{views}. 
%
 Hotstuff follows this pattern and consists of three voting rounds, one for the non-equivocation phase, and two for the persistence phase. In the \textit{Prepare round}, the leader proposes a block $A$ with view $v$ at position $i$ and each replica votes to prepare if it has not already prepared a block at $i$ with a higher view.  If the leader successfully obtains a QC ($n-f$ distinct replica votes) in the Prepare round, a \textit{prepareQC} forms, and the leader moves on to the Pre-Commit round. The existence of a prepareQC ensures agreement on $A$: no other block could have been certified for view $v$ as any two prepareQC must overlap in at least one honest replica. 
In the Pre-Commit round, the leader broadcasts the prepareQC; replicas vote to accept the prepareQC and locally update their highest prepareQC. The leader waits to receive $n-f$ Pre-Commit votes to form a precommitQC. 
In the Commit round, the leader  broadcasts a final \textit{precommitQC}. Replicas become \textit{locked} on this QC: they will never vote for a conflicting block unless they receive a prepareQC in a higher view.  The existence of a higher conflicting prepareQC is evidence that the locked QC could not have committed (honest replicas would not vote to support two conflicting blocks). Finally, the leader forms a commitQC upon receiving $n-f$ Commit votes. It attaches the commitQC in a Decide round to inform replicas that the block committed, at which point they can execute the block's operations and move to the next view. 
We remark that the Pre-Commit round in HotStuff is necessary only for liveness rather than safety (the \textit{hidden lock problem}~\cite{yin2019hotstuff}). 
Recent BFT protocols manage to avoid the Pre-Commit round of Hotstuff (thus achieving optimal two-round finality) by, respectively, eschewing linear world complexity \cite{jalalzai2020fasthotstuff, jolteon-ditto}, elongating view changes \cite{sui2022marlin}, or introducing novel cryptography \cite{giridharan2021no}.



\subsection{Chaining and LSO}
\label{s:hs-cslo}

\par \textbf{Chaining.} The aforementioned protocol takes three phases to commit a block: each operation thus requires forming three individual QCs. Prior work observes that, while each step serves a different purpose, all have identical structure: the leader proposes a block, collects votes and forms a QC. To amortize cryptographic costs and minimize latency,  one may pipeline commands such that a single QC simultaneously serves as  $prepareQC$, $precommitQC$, and $commitQC$ for different
blocks. The number of QCs in the chain indicates whether a block has been prepared or committed. 

\begin{figure}[t]
    \centering
    \includegraphics[scale=0.4]{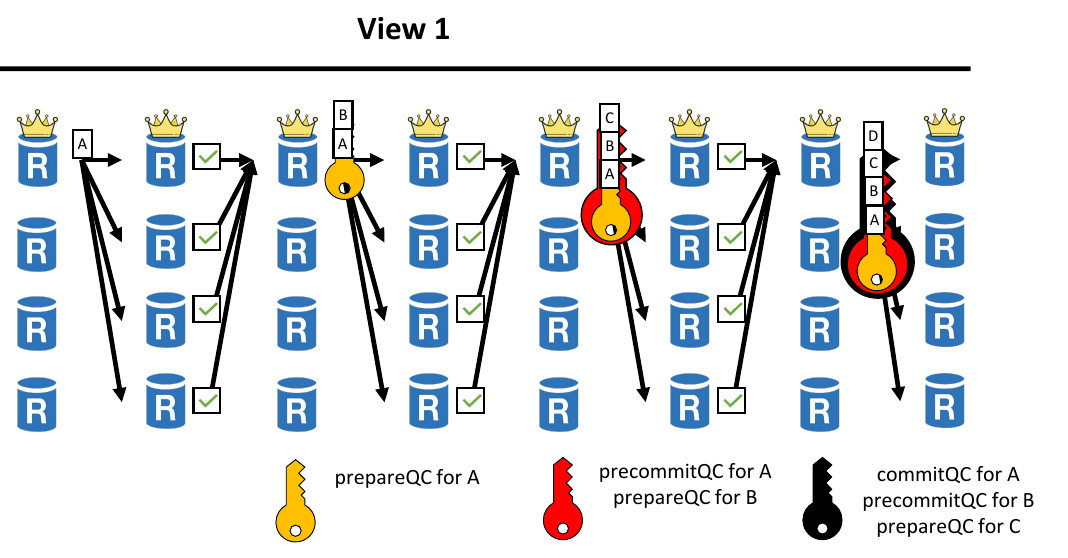}
    \caption{Chained Hotstuff \nc{If decide to remove previous figure, need to make sure we put legend on this graph}}
    \label{fig:chaining} 
\end{figure}

Consider for example a scenario in which a chained protocol is attempting to commit four blocks $A$, $B$,
$C$, and $D$ (Figure~\ref{fig:chaining}). The view leader first proposes $A$, collects $n-f$ votes for $A$, and forms its first QC ($QC_{A}$). $QC_{A}$ is the
\textit{prepareQC} for $A$. Next, the leader proposes $B$, indicating that (i) $A$ is the \textit{parent block} of $B$, and (ii) that $A$ has been certified by QC $QC_{A}$. It once again collects $n-f$ votes, forming $QC_{B}$. This QC acts as both the 
\textit{precommitQC} for $A$ and the \textit{prepareQC} for $B$. Similarly, the leader proposes and obtains a QC for block $C$. It marks $C$ as the parent block and uses $QC_{C}$ to attest to the validity of both $B$ and $A$ (implicitly through $B$'s backpointer). It forms a $QC_{C}$ which acts as a \textit{commitQC} for $A$, \textit{precommitQC} for $B$, and \textit{prepareQC} 
for $C$. As three \textit{consecutive} QCs now attest to $A$'s presence in the chain, $A$ is now committed. In a fourth step, the leader certifies $C$ using $QC_{C}$ and proposes $D$. Upon receiving this QC, replicas learn that $A$ has been committed and can thus safely execute operations in the block.

%

\par \textbf{LSO.} In the previous example, a single leader drives the full protocol (a \textit{stable leader}). It is responsible for deciding which block to include next in the chain, for collecting replica votes, for creating the corresponding QC and broadcasting it to all replicas. Many BFT protocols adopt this paradigm and do not rotate leaders in the absence of failures~\cite{aardvark,kotla10zyzzyva, castro1999pbft, gueta2019sbft} among others. This raises fairness concerns; malicious leader can censor operations, penalize specific users or influence operation ordering \cite{daian2019flash}. A plethora of recent work addresses order-related fairness issues in BFT systems \cite{zhang2020, kelkar2020order, kelkar2021themis, Bullshark, suri2021basil, danezis2022narwhal, malkhi2022maximal, heimbach2022sok}.
While fairness concerns in chained BFT protocols can be mitigated by carefully optimizing the leader-selection process using leader reputation schemes like Carousel~\cite{carousel}, the problem cannot be fully side-stepped.
In the \textit{leader-speak-once} (LSO) model, each view lasts for the duration of a single protocol phase, thus minimizing each leader's influence on new proposals. A leader receives votes, forms a QC, proposes the next block, and is immediately rotated out. Upon receiving the block proposal, replicas directly increment their view and send their votes to the \textit{next} leader in the rotation. 


\subsection{Liveness Concerns}
\label{s:live-kaput}

As we saw in Section~\ref{s:hs-basic}, Hotstuff commits a block once three consecutive QC attest to the block's validity. As such, it requires a sequence of four consecutive leaders (three to form a QC and a fourth to disseminate the final QC). While other CLSO protocols reduce the number of consecutive QCs to two~\cite{diem,jalalzai2020fasthotstuff}, all require that these QCs form in consecutive views. Committing a block thus requires a sequence of $k=3$ or $k=4$ honest leaders. Unfortunately, we find that requiring consecutive honest leaders introduces a significant liveness issue where blocks can be prevented from committing for long periods of time. All existing CLSO protocols suffer from this limitation. 
\begin{figure}[h!]
\begin{minipage}{.45\textwidth}
    \centering
    \includegraphics[scale=0.4]{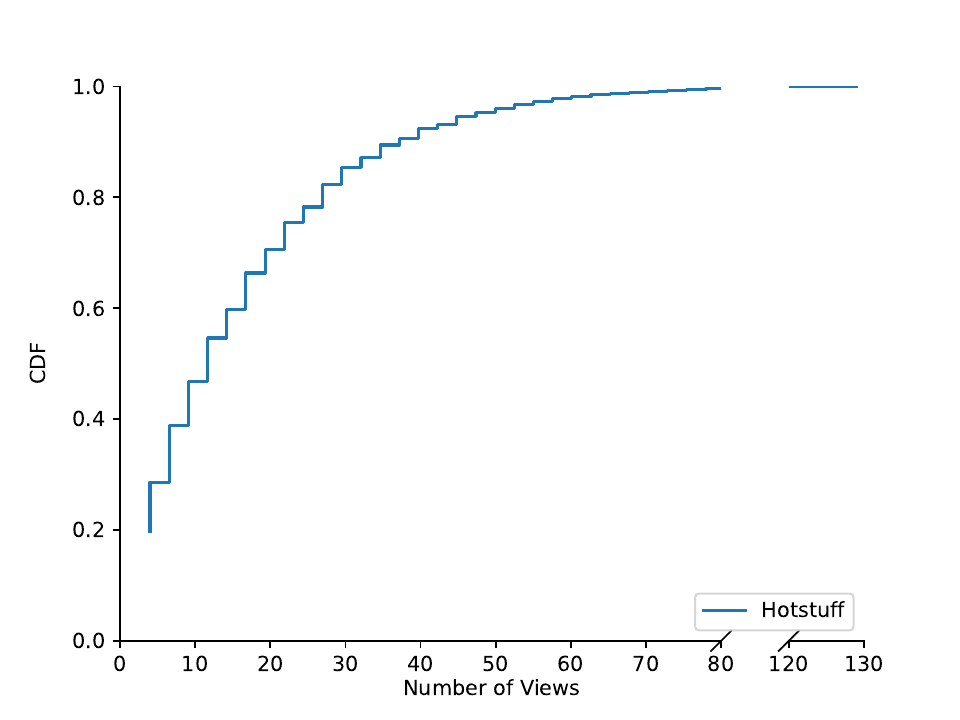}
    \caption{\# number of views needed for commit}
    \label{fig:faulty} 
    \end{minipage}
     \begin{minipage}{.45\textwidth}
    \centering
    \includegraphics[scale=0.4]{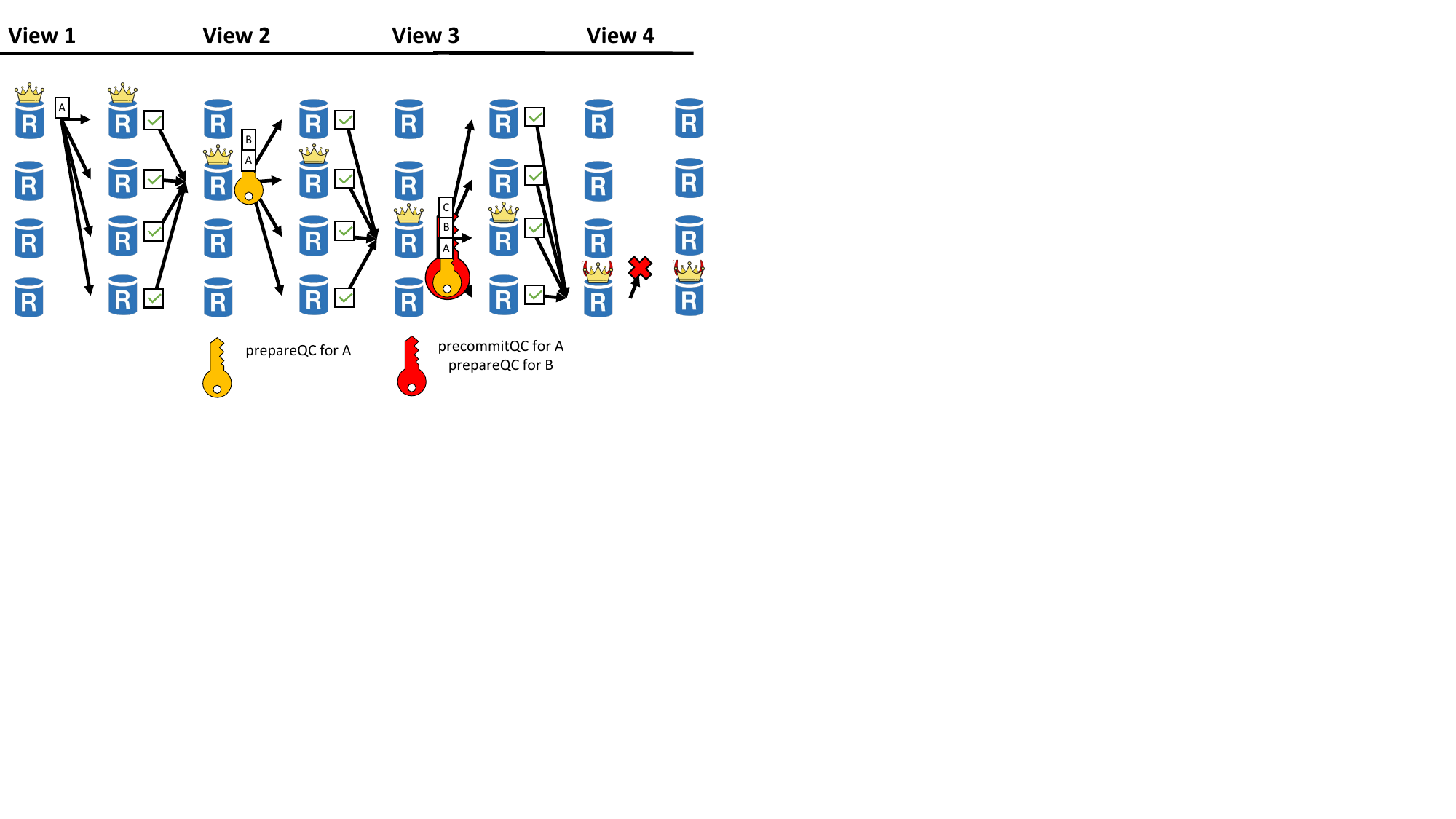}
    \caption{Liveness issue with CLSO protocols}
    \label{fig:liveness}
    \end{minipage}
\end{figure}

Consider replicas $R_1$, $R_2$, $R_3$ and $R_4$, with $R_4$ being faulty (Figure~\ref{fig:liveness}).
Leaders are elected round-robin.
If leader $R_1$ initially proposes block $A$, $R_4$ might never broadcast the final $QC$ (acting as \textit{commitQC}) necessary for replicas to execute $A$. Likewise, if $R_2$ and $R_3$ propose blocks $B$ and $C$ respectively, $R_4$ may fail to form the corresponding \textit{precommitQC} (for $B$) and \textit{prepareQC}(for $C$). 
In fact, with four replicas, a single faulty leader can prevent Hotstuff from committing \textit{any} block! More generally, requiring $n$ consecutive leaders
can create significant latency spikes, even for large participant sets, as blocks must be re-proposed until they find a sequence of n consecutive leaders. We measure this danger through simulation in Figure~\ref{fig:faulty}, where we calculate the number of phases necessary to commit a block given a random assignment of faulty nodes and random leader election policy. In our experiment, Hotstuff requires an average of 12 phases to commit (a three-fold increase over the failure-free case), and has an observed worst-case latency of 129 rounds. We note that, in the absolute worst case, Hotstuff (or any CLSO protocol with $k=4$) may \textit{never} commit a block.

\subsection{It's not easy to strengthen liveness}
\label{s:relax-take-it-eeeaaasy}
Requiring $k+1$ consecutive honest leaders to guarantee block commitment stems from a CLSO's protocol need to commit a block if and only if $k$ \textit{consecutive} QCs certify it.
Weakening this requirement leads to safety violations ~\cite{bano2020twins}). To illustrate, let us assume that any two (possibly non-contiguous) QCs certifying
block $B$ suffice to commit $B$ (we adopt the more efficient "two-QC" rule of Jolteon~\cite{jolteon-ditto} or Fast Hotstuff~\cite{jalalzai2020fasthotstuff} for simplicity, but the same reasoning holds for Hotstuff's three-consecutive QC rule). We assume, as in HotStuff, that replicas vote to certify a block as long as they are not locked on a higher conflicting block. This is true for all voting in our example. We further assume that leaders are elected round-robin where $R_i$ is leader for view $v$ where $v \pmod i==0$. We can easily show that temporary periods of asynchrony will lead conflicting blocks to commit.
\begin{itemize}[leftmargin=*]
    \item \ul{\textit{Views 1-2}}. $R_1$ proposes $A_1$. A QC forms
    for $A_1$ at $R_2$. $R_2$'s broadcast messages are delayed. Asynchrony leads to a view change.
\item \ul{\textit{View 3-5}}. $R_3$ receives responses from all replicas expect for $R_2$. All responses are empty (recall that replicas only include QCs in view changes, not votes). $R_3$ proposes $B_1$. A $QC$ forms for $B_1$ at $R_4$. $R_4$'s broadcast messages are delayed. Asynchrony leads to view change.
\item \ul{\textit{View 6-8}}. $R_1$ receives responses from all replicas expect for $R_4$ and learns about $QC_{A_1}$. It proposes $A_2$, which extends $A_1$. A QC forms at $R_2$. $R_2$ sees that two QCs certify $A_1$ and thus commits $A_1$. Asynchrony leads to a view change.
\item \ul{\textit{View 9-11}}. $R_3$ receives responses from all replicas expect for $R_2$ and learns about $QC_{B_1}$. It proposes $B_2$, which extends $B_1$. A QC forms at $R_4$. $R_4$ sees that two QCs certify $B_1$ and thus commits $B_1$. 
\end{itemize}
$R_3$ and $R_2$ have committed conflicting blocks, violating safety. 
The root cause here is simple: committing a proposal after observing QCs in non-contiguous views is dangerous because there may exist a higher conflicting QC.
 In contrast, requiring QCs to be in contiguous views ensures that, for any committing proposal $p$, a QC that extends p will be preserved across view changes: 
 Since at least two QC's are required to commit $p$ (three for HotStuff), existence of the latest $QC_{lat}$ implies that at least $f+1$ honest replicas have observed (at least) the preceding $QC_{pre}$. Since $QC_{pre}$ (by assumption of being contiguous) has the highest view (bar $Q_{lat}$), it follows that every view change (a quorum of $n-f$) must observe $QC_{pre}$ at least once. Thus, all future proposals must extend $p$.

This paper asks: can we strengthen liveness (\nonconsec{}) without violating safety? 
We answer in the affirmative. We introduce \sys{}, CLSO protocol that, after GST, will commit all blocks proposed by an honest leader in view $v$ after at most two honest (possibly non-contiguous) views $v'$ and $v''$.
\section{\sys{}}\label{sec:protocol}
 \sys{} achieves optimal phase complexity (two phases to commit a block), quadratic word complexity with threshold signatures, linear word complexity with SNARKs, and responsiveness with consecutive honest leaders. Specifically \sys{} satisfies the following property. 

\begin{theorem} (\nonconsec{}). \label{siesta-main}
After GST, if an honest leader in view $v$ proposes a value, then it is guaranteed to commit if views $v_{i_1}<v_{i_2}<...<v_{i_k}$ have honest leaders (non-consecutive).
\end{theorem}

This property has consequences for both safety and liveness. For safety, \sys{} must ensure that, in the presence of an honest leader $L_1$ proposing $B$, the existence of two QCs in non-contiguous views for $B$ be sufficient to guarantee that no conflicting block $B$' can commit. This theorem also places stringent liveness requirements on \sys{}: after GST, all blocks proposed by an honest leader \textit{must} be committed. Note that this is not a property that is traditionally guaranteed by CLSO protocols. In the rest of this section, we first describe the core intuition behind \sys{} before describing the protocol in more detail.


\subsection{\sys{} Overview}

\par \textbf{The case of the conflicting QC.} 
As shown in \S\ref{s:relax-take-it-eeeaaasy}, committing a block requires ensuring that no higher conflicting QC could have formed. Committing with consecutive QCs 
guarantees precisely that (\S\ref{s:relax-take-it-eeeaaasy})..
To satisfy \nonconsec{}, a protocol must instead commit blocks even when the QCs certifying them are not consecutive. We thus require alternative mechanisms to prevent conflicting QCs from forming. In an ideal world, one would design a clever algorithm that prevents \textit{all} such QCs from forming. Unfortunately, this is impossible~\cite{good-case}. Instead, \sys{} proceeds in two ways: under asynchrony and omission faults, \sys{} does indeed \textit{preclude} all conflicting QCs. In the presence of equivocation, \sys{} instead reliably \textit{detects} when a conflicting QC could have formed  and immediately aborts committing this block. Together, these mechanisms ensure that, after GST, all blocks proposed by honest leaders will eventually commit.

\par \textbf{Key Idea}. \sys's key insight is simple. It explicitly makes use of information that all other CLSO protocols (and most other BFT protocols) traditionally discard after processing. 
\textit{\phaseone{}} messages (\textit{Pre-prepare} messages in PBFT)\footnote{to avoid naming conflicts, we will refer to these messages as \textit{Prepare} messages}. These messages have until now only been used to achieve an optimistic fast 
path (Zyzzyva~\cite{kotla10zyzzyva} and SBFT~\cite{gueta2019sbft}) in which a superquorum containing all $n$ replicas 
informs the client that the block has been persisted. In non-fast path protocols, these messages were thought to convey no useful information as they precede the protocol's non-equivocation phase, and thus hold no bearing to maintaining \textit{safety}.
\sys{} instead uses them to improve \textit{liveness} by reliably distinguishing between asynchrony/omission faults and equivocation. 
\par \textbf{Technical Intuition}. By quorum intersection, if a QC forms for a block $B$, all subsequent view change leaders will receive at least one \textit{Prepare} message for $B$. By extending this block, subsequent leaders will never create conflicting blocks. If two conflicting QCs do form as a result of equivocation, subsequent leaders will necessarily observe the existence of two \textit{Prepare} messages {in the same view}, and thus abort block commitment. These conflicting messages further create explicit evidence of misbehavior, allowing the faulty leader to be removed.

While the above approach elegantly guarantees safety without requiring consecutive QCs, it does not yet fully satisfy Theorem \ref{siesta-main}. 
For example, consider the following scenario with a sequence of five leaders: $L_1$ (honest), $L_2$ (faulty), $L_3$ (honest), $L_4$ (faulty), and $L_5$ (honest). $L_1$ could propose a block $B$ (after GST), all honest replicas vote for it, implicitly forming a QC for $B$.  $L_2$ could, however, fail to assemble and disseminate this QC. $L_3$ would not observe a QC for $B$ and instead propose a new block $B'$ that extends $B$. Similarly, $L_4$ would fail to generate and disseminate a QC, and $L_5$ would fail to observe a QC for $B'$. To satisfy Theorem \ref{siesta-main} we must commit $B$ since there were three honest leaders; however, in this scenario we fail to do so. To address this issue, \sys{} develops a novel technique, QC materialization, that makes these implicit QCs explicit, allowing replicas to commit the relevant blocks. QC materialization hinges on two observations: after GST, if an honest leader broadcasts a \textit{Prepare} message, all honest replicas will receive it and send a reply. Second, an honest leader is guaranteed to receive replies from all honest replicas in $\Delta$ time (after GST).

\par \textbf{Protocol Structure.}\label{subsec:datastructures}
\sys{} shares the same structure as other CLSO protocols. It consists of four components: a fast view change, a slow view change, a commit procedure, and a view synchronizer. Fast view changes occur in the absence of delayed messages or failures. Slow view changes are triggered by lack of progress (view timeouts). For each view, the leader runs a commit procedure to determine which blocks in the chain can safely be committed. The view synchronizer ensures all honest replicas remain in the same view for sufficient amount of time.  \sys{} adopts the same view synchronizer as prior work~\cite{naor2020cogsworth,bravo2020making}; we focus on the other components here.

\subsection{\sys{} Data Structures}
%
\par \textbf{Blocks and Block Format.} As is standard, \sys{} batches client requests into \textit{blocks}, with each block containing
a hash pointer to its parent block (or to null in the case of the \textit{genesis} block). A block's position in the
chain is known as its height. A block $B \coloneqq \langle v, p, QC, b, NV_{set} \rangle$ contains the following information: $v$, the view the block was
proposed in; $p$, its parent block; $QC$, the quorum certificate certifying an
ancestor 
block (more details later); and $b$, a batch of client transactions (Alg. \ref{alg:prelim} 
lines \ref{prelim:block:begin}-\ref{prelim:block:end}). Additionally, blocks proposed in the slow view change
must contain $NV_{set}$ (Alg. \ref{alg:prelim} line \ref{prelim:slow}), the set of \textsc{New-view} messages (more 
detail later). A block $B$ is
valid if 1) its parent block is valid (or $B$ is genesis)
2) all included client transactions $b$
satisfy all application level validity predicates, and 3) all included signatures are valid. Honest replicas only accept valid blocks -- we omit validation checks from the pseudocode.

\par \textbf{Block extension and conflicts}. Parent pointers link blocks into a chain. We define a block ancestor of $B$ to be any block $B_{anc}$ for which a path (of parent links) exists from $B$ to $B_{anc}$.
We say
a block $B'$ extends (or is descendant of) a block $B$ ($B \longleftarrow B'$) if $B$ is an ancestor of $B'$. 
We say that $B'$ conflicts with $B$ if neither extends the other ($\neg (B \longleftarrow B' \lor B' \longleftarrow B)$).
Informally, if $B'$ conflicts with $B$, these blocks are on separate forks of the chain and only one of them can commit. By convention, we say that blocks extend themselves. 

\par \textbf{Equivocation.} Honest leaders may propose only a single block per view. We label conflicting blocks with the same view as \textit{equivocating}. An equivocation proof (more details later) constitutes evidence of leader equivocation.

\par \textbf{Message Types.} In \sys{} there are three types of messages: \textsc{Vote-req}, \textsc{Vote-resp}, and \textsc{New-view}. \msg{Vote-req}{B} messages are the \textit{Proposal} messages in \sys{} and contain $B$, the leader's proposed block. Replicas send \msg{Vote-resp}{B} messages to vote on a proposal for $B$.
Since blocks are
chained together, a \textsc{Vote-resp} message for a block counts also as \textsc{Vote-resp} for all of its ancestors. 
Each replica stores its current view, $v_r$, the latest accepted (highest view) \textsc{Vote-req}, $vr_r$, and \textsc{Vote-resp}, $vp_r$ (Alg \ref{alg:normal-case}. lines \ref{fast-local-begin}-\ref{fast-local-end}). 
\msg{New-view}{v,vr_r,vp_r} messages are used by the slow view change to maintain progress despite failures or asynchrony. They contain the view $v$ that the replica is advancing to, the replica's latest \textsc{Vote-req}, $vr_r$, and its latest \textsc{Vote-resp} $vp_r$. 


\par \textbf{Quorum Certificates}. A $QC \coloneqq \langle Q, B \rangle$ consists of a set $Q \coloneqq n-f$ \textsc{Vote-resp} messages and a \textit{certified} block $B$.
We say a block $B$ is \textit{certified} if there exists a quorum $Q$ of $n-f$ \textsc{Vote-resp} messages for $B$ itself (direct) or a descendant block $B'$ (indirect). 
Given a set of any $n-f$ \textsc{Vote-resp} messages, we can thus determine which block was certified by identifying the highest (w.r.t view) common ancestor (Alg. \ref{alg:prelim} line \ref{prelim:anc}).
A QC contains $Q$ and $B \coloneqq B_{anc}$, the highest block that $Q$ certifies.
(Alg. \ref{alg:prelim} lines \ref{prelim:qcblock}-\ref{prelim:qcvotes}).
We say that two QCs conflict if they certify blocks that conflict. 
In the rest of the paper, we, 
denote a QC as \textit{implicit} as soon as the necessary \textsc{Vote-resp}s to form $Q$ are cast, but the QC has not yet been assembled. A leader materializes an implicit QC into an explicit QC by assembling the necessary votes.


\par \textbf{Ranking.} We introduce a notion of ranking rules for both blocks and QCs. 
Blocks with higher views have higher ranks; ties are broken by the rank of their contained QCs (Alg. \ref{alg:prelim} line \ref{prelim:blockrank}). These rules are used to determine whether a block is safe to accept in the slow view change.

\begin{algorithm}

\begin{algorithmic}[1]
\Procedure{CreateBlock}{$v,B_p,QC,NV_{set}$}
\State{$B.v\gets v$}\Comment{The view the block $B$ is proposed in}\label{prelim:block:begin}
\State{$B.p\gets B_p$}\Comment{The parent block $B$ extends}
\State{$B.QC_{anc}\gets QC$}\Comment{The QC for an ancestor block that $B$ contains}
\State{$B.b\gets$ a batch of client transactions}\label{prelim:block:end}
\State{$B.NV_{set}\gets NV_{set}$}\Comment{The set of \textsc{NewView} messages for slow blocks}\label{prelim:slow}
\State{\Return $B$}
\EndProcedure

\Procedure{BlockRank}{$B_1,B_2$}\Comment{Blocks are ranked by view, ties are broken by the higher QC}
\State{\Return $B_1.v \geq B_2.v \land \Call{QCRank}{B_1.QC_{anc},B_2.QC_{anc}}$}\label{prelim:blockrank}
\EndProcedure

\Procedure{CreateQC}{$Q$}\Comment{$Q \coloneqq n-f$ \textsc{Vote-resp} messages}
\State{$B_{anc} \gets $ highest common ancestor block of $Q$}\label{prelim:anc}
\State{$QC.B\gets B_{anc}$}\Comment{The block the QC certifies}\label{prelim:qcblock}
\State{$QC.Q\gets Q$}\label{prelim:qcvotes}
\State{\Return $QC$}
\EndProcedure

\Procedure{QCRank}{$QC_1,QC_2$}\Comment{QCs are ranked by view of the block they certify}
\State{\Return $QC_1.B.v>QC_2.B.v$} 
\EndProcedure

\end{algorithmic}

\caption{Data Structure Utilities}\label{alg:prelim}
\end{algorithm}

\subsection{Protocol Details}

\textbf{Fast View Change (FVC) .} We first focus on the steady state. Successive leaders transmit state through a fast view change.  The structure of \sys{}'s is identical to existing CLSO protocols in that there are two steps: 1) the leader proposes a valid block to all
replicas (sending step) and 2) replicas accept the block and forward their vote to the leader of the next view (receiving step).
\par \ul{\textit{Sending step.}}The leader of view $v_r+1$ 
forms a valid QC for a block $B$ in view $v_r$ when it receives receives $n-f$ matching votes for $B$ (Alg. \ref{alg:normal-case} lines \ref{fast-qc-start}-\ref{fast-qc-end}).
As the views are contiguous, proposing a block that extends $B$ will not result in a conflicting QC.
The leader can then safely propose a new block $B'$, which extends $B$ (Alg. \ref{alg:normal-case} lines \ref{fast-propose-start}-\ref{fast-propose-end}).
\par \ul{\textit{Receiving step.}}
Replicas deem a proposal for $B'$ valid
if the associated QC is for $v_{r}$ (i.e. contiguous), and $B'$ extends $B$ (Alg. \ref{alg:normal-case} line \ref{fast-receive}). 
It updates its current view $v_r \coloneqq v_r + 1$, its latest received $vr_r \coloneqq $ \msg{Vote-req}{B'} 
and its latest sent  $vp_r \coloneqq $\msg{Vote-resp}{B'}, 
indicating its support for block $B'$ (Alg. \ref{alg:normal-case} lines \ref{fast-update1-begin}-\ref{fast-update1-end}). It then sends $vp_r$ to the leader of the next view ($v_r +1$) (Alg. \ref{alg:normal-case} line \ref{fast-send}). 

\begin{algorithm}
\begin{algorithmic}[1]
\State{$v_r\gets 1$}\Comment{Current view of replica $r$}\label{fast-local-begin}
\State{$vr_r\gets \bot$}\Comment{Latest \textsc{Vote-req} received}
\State{$vp_r\gets \bot$}\Comment{Latest \textsc{Vote-resp} sent}\label{fast-local-end}

\State{$//$ \textit{only the leader of view $v_r+1$}}
\Event{receiving $S\gets 2f+1$ \textit{matching} \msg{Vote-resp}{B} messages while in view $v_r$}\label{fast-qc-start}
\State{$QC_B\gets \Call{CreateQC}{S}$}\Comment{Certify $B$ since it has $2f+1$ votes}\label{fast-qc-end}
\State{$B'\gets \Call{CreateBlock}{v_r+1,B,QC_B,\bot}$}\Comment{Propose $B'$ which has $B$ as its parent}\label{fast-propose-start}
\State{\textbf{send} \msg{Vote-req}{B'} to all}\label{fast-propose-end}
\EndEvent

\Event{receiving a valid \msg{Vote-req}{B'} from $L_{v_r+1}$}
\State{$QC_B\gets B'.QC_{anc}$}
\State{$B\gets QC_B.B$}\Comment{Gets the certified block that $B'$ extends}
\State{$//$ \textit{in the normal case the QC will be from the previous view}}
\If{$B'.NV_{set}=\bot \land B.v+1=B'.v \land B=B'.p$}\label{fast-receive}
\State{$v_r\gets B'.v$}\Comment{Move to the next view}\label{fast-update1-begin}
\State{$vr_r\gets$ \msg{Vote-req}{B'}}\Comment{Update latest \textsc{Vote-req}}
\State{$vp_r\gets$ \msg{Vote-resp}{B'}}\Comment{Update latest \textsc{Vote-resp}}\label{fast-update1-end}
\State{\textbf{send} $vp_r$ to $L_{v_r+1}$}\Comment{Send vote to the next leader}\label{fast-send}
\EndIf
\EndEvent
\end{algorithmic}

\caption{Fast View Change (Steady State)}\label{alg:normal-case}
\end{algorithm}

\textbf{Slow View Change (SVC).} The FVC in \sys{} is simple: as views are contiguous, the new leader is guaranteed to see the latest possible QC. It can then easily extend the chain without any risking of a conflicting QC forming. There is no such continuity in the slow view change, which requires more care. The SVC has two main objectives: 1) maintain consistency \textit{across} views and 2) continue making progress on honest proposed blocks.
\par \ul{\textit{Local State.}} For the slow view change, each replica maintains a view timer that resets every time it advances to a new view. This timer is used to detect a lack of progress in a view. The leader of the new view additionally maintains 1) $NV_{set}$, the set of \textsc{New-view} messages received, 2) $B_{parent}$, the parent block of the new leader's next proposal, 3) $QC_{anc}$, the highest ranked explicit QC that $B$ extends, and 4) a materialization timer (Alg. \ref{alg:siesta-vc} lines \ref{local-begin}-\ref{local-end}) (more detail follows).
\par \ul{\textit{Trigger Conditions.}} A slow view change is triggered when enough replicas fail to make progress in the current view (when their view timer expires). A replica then indicates that it wants to change views by sending the next leader a \textsc{New-view} message containing its relevant local state (Alg. \ref{alg:siesta-vc} lines \ref{trigger-begin}-\ref{trigger-end}). When the new leader receives a \msg{New-view}{v,vr_r,vp_r} message, it adds it to the set of \textsc{New-view} messages received so far for the view $v$ (Alg. \ref{alg:siesta-vc} lines \ref{newview-begin}-\ref{newview-end}). A slow view change is  triggered when sufficiently many ($n-f$) \textsc{New-view} messages have been received (Alg. \ref{alg:siesta-vc} line \ref{quorumreached}).

\par \ul{\textit{Parent Block Selection.}} The new leader first selects a parent block $B_{parent}$ to extend. Recall that in \sys{}, unlike in other CLSO protocols, \textsc{New-view} messages include a replicas' last seen \textsc{Vote-req} message. The leader then always selects the highest ranked block among these \textsc{Vote-req} messages (Alg. \ref{alg:siesta-vc} line \ref{highrank}). 
In doing so, the leader guarantees that it always extends the latest block for which a QC \textit{could} have formed (but that the leader did not necessarily receive). By the quorum intersection property, if forming a QC requires at least $n-f$ replicas receiving the corresponding \textsc{Vote-req} messages, at least one of these messages would have been included in the $n-f$ \textsc{New-view} messages.
In the absence of explicit equivocation, using \textsc{Vote-req} messages in this way precludes the leader from extending a block that conflicts with a QC in a higher view.
If a previous leader does equivocate, there may exist \textsc{Vote-req} messages for equivocating blocks that have the same (highest) rank. The leader will pick one of these equivocating blocks to extend arbitrarily, which may result in the formation of a conflicting QC.
We discuss how \sys{} safely handles this scenario in the commit rule.

\par \ul{\textit{Implicit QC Materialization}.} 
Next, the new view leader must ensure that, after GST, any block proposed by an honest leader will eventually be committed. \sys{} leverages \textsc{Vote-req} messages to enforce this invariant through a novel QC materialization technique. \sys{} makes three observations: 1) after GST, all honest replicas are guaranteed to vote in favor of an honest leader's proposal. 2) after GST, the next leader is guaranteed to receive responses from all honest nodes within its timeout $\Delta$ 3) before GST, there is no requirement to eventually commit blocks proposed by honest leaders. It follows that, after GST, if an honest leader proposed a block $B$, an implicit QC formed and subsequent leaders will necessarily receive $n-f$ \msg{Vote-resp}{B} (or descendants of $B$). As such, any time a leader sees $n-f$ \msg{Vote-resp}{B'} messages for some block $B'$, such that $B \longleftarrow B'$, it could have been proposed by an honest leader and must therefore be certified. Note that \sys{} enforces this guarantee for liveness, not safety. Before GST, honest leaders' proposals may - as is the case in existing CLSO protocols -- fail to generate a QC.

The leader first identifies the highest ranked $QC_{anc}$
(Alg. \ref{alg:siesta-vc} line \ref{qcrank}) on the chain that certifies an ancestor of $B_{parent}$, just as one would in traditional CLSO protocols. Next, the leader tries to materialize any higher ranked implicit QCs on the chain. If there are enough \textsc{Vote-resp} messages to materialize a QC for $B_{parent}$ (Alg. \ref{alg:siesta-vc} lines \ref{cancel-mat-begin}-\ref{cancel-mat-end}), the leader materializes this QC and directly proposes a new block that extends $B_{parent}$. This is safe as $B_{parent}$ is (by block selection) necessarily the highest ranked block on the chain. If, instead, there are insufficient \textsc{Vote-resp} messages, the leader starts a materialization timer during which it waits for additional \textsc{New-view} messages in order to materialize implicit QCs for descendants of $QC_{anc}.B$.
In line with our aforementioned observations, the materialization timer must be greater than or equal to $\Delta$~\cite{jolteon-ditto} in order to guarantee that, after GST, the messages of honest nodes will all be received. If the leader eventually receives $n-f$ \msg{Vote-resp}{B_{desc}} where $B_{desc}$ is a descendant of $QC_{anc}.B$, $B_{desc}$ could have been proposed by an honest node. The leader thus materializes an explicit QC for $B_{desc}$ and updates its local knowledge of the highest ranked known $QC_{anc}$ (Alg. \ref{alg:siesta-vc} lines \ref{check-update}-\ref{update-qc}). If, while waiting, the leader receives $n-f$ \msg{Vote-resp}{B_{parent}}, the leader instead updates $QC_{anc}$ certify $B_{parent}$, and the materialization timer can be canceled. 
Finally, thes leader proposes a new block $B$ with parent block $B.p \coloneqq B_{parent}$, QC $B.QC \coloneqq QC_{anc}$,
and $B.NV_{set} \coloneqq NV_{set}[v]$ the set 
of \textsc{New-view} messages received (Alg. \ref{alg:siesta-vc} lines \ref{propose1-start1}-\ref{propose1-end1}, \ref{propose2-start2}-\ref{propose2-end2}). 

\par \ul{\textit{View Change Validation.}} When a replica receives a valid \msg{vote-req}{B'} proposal from the leader, the replica checks that the leader did in fact perform the view change correctly (Alg. \ref{alg:siesta-vc} line \ref{verify-correct-vc}). It confirms that 1) the leader obtained $n-f$ \textsc{New-view} messages 2) that the proposed block's parent was in fact the highest ranked blocks among \textsc{Vote-req} messages 3) that the proposal extends the highest QC received by the leader. When confirmed, the replica updates its state (Alg. \ref{alg:siesta-vc} lines \ref{state-update-begin}-\ref{state-update-end}), resets its view timer (Alg. \ref{alg:siesta-vc} line \ref{reset}) and sends a \textsc{Vote-resp} to the next leader (Alg. \ref{alg:siesta-vc} line \ref{send-vote}).

\begin{algorithm}
\begin{algorithmic}[1]
\State{view timer $\gets 5\Delta$ delay}\Comment{Set view timer delay}\label{local-begin}
\State{materialization timer $\gets \Delta$ delay}\Comment{Set materialization timer delay}
\State{$NV_{set} \gets \{\}$}\Comment{Stores \textsc{New-view messages}}
\State{$B_{parent} \gets \bot$}\Comment{Highest ranked \textsc{Vote-req} block}
\State{$QC_{anc} \gets \bot$}\Comment{The highest ranked explicit QC}
\State{$B_{anc} \gets \bot$}\Comment{The block $QC_{anc}$ certifies}\label{local-end}

\Event{view timer for $v_r$ expiring}\label{trigger-begin}
\State{\textbf{send} \msg{New-view}{v_r+1,vr_r,vp_r} to $L_{v_r+1}$}\label{trigger-end}
\EndEvent

\Event{receiving \msg{New-view}{v,vr_r,vp_r} for view $v>v_r$}\label{newview-begin}
\State{$NV_{set}[v] \gets NV_{set}[v] \cup $ \msg{New-view}{v,vr_r,vp_r}}\label{newview-end}
\If{$|NV_{set}[v]| = n-f$}\Comment{Slow view change triggered}\label{quorumreached}
\State{$B_{parent}\gets \Call{HighVoteReq}{S}$}\Comment{Finds the highest ranked block to extend}\label{highrank}
\State{$QC_{anc}\gets B.QC_{anc}$}\Comment{Gets the QC contained within $B$}\label{qcrank}
\State{$B_{anc}\gets QC_{anc}.B$}
\State{\textbf{start} materialization timer}\label{mattimer-start}
\EndIf
\State{\textit{// while waiting for materialization timer, continually check to see if $QC_{anc}$ can be updated}}
\If{$NV_{set}[v]$ \textit{contains} $n-f$ \msg{Vote-resp}{B_{desc}} where $B_{desc}$ \textbf{extends} $B_{anc}$}\label{check-update}
\State{$QC_{anc}\gets \Call{CreateQC}{NV_{set}[v]}$} \Comment{Update $QC_{anc}$ to be newly formed QC}\label{update-qc}
\State{$B_{anc}\gets QC_{anc}.B$}
\EndIf
\If{$B_{anc}=B_{parent}$}\Comment{$QC_{anc}$ has highest possible rank, propose a new block}\label{cancel-mat-begin}
\State{\textbf{cancel} materialization timer}\label{cancel-mat-end}
\State{$B'\gets \Call{CreateBlock}{v,B_{parent},QC_{anc},NV_{set}[v]}$}\Comment{Propose $B'$}\label{propose1-start1}
\State{\textbf{send} \msg{Vote-req}{B'} to all}\label{propose1-end1}
\EndIf
\EndEvent

\Event{materialization timer expiring}\label{expire}
\State{$B'\gets \Call{CreateBlock}{v,B_{parent},QC_{anc},NV_{set}[v]}$}\Comment{Propose $B'$, with $B$ as its parent.}\label{propose2-start2}
\State{\textbf{send} \msg{Vote-req}{B'} to all}\label{propose2-end2}
\EndEvent



\Event{receiving a valid \msg{Vote-req}{B'} from $L_{v_r+1}$}
\State{$QC_B\gets B'.QC_{anc}$}
\State{$B\gets QC_B.B$}
\State{$//$ \textit{verify that $L_{v_r+1}$ did the view change correctly}}
\If{$|B'.NV_{set}|\geq 2f+1 \land \Call{HighVoteReq}{B'.NV_{set}}=B'.p \land B'$ \textit{extends} $B$}\label{verify-correct-vc}
\State{$v_r\gets B'.v$}\Comment{Move to the next view}\label{state-update-begin}
\State{$vr_r\gets$ \msg{Vote-req}{B'}}\Comment{Update latest \textsc{Vote-req}}
\State{$vp_r\gets$ \msg{Vote-resp}{B'}}\Comment{Update latest \textsc{Vote-resp}}\label{state-update-end}
\State{\textbf{reset} view timer}\label{reset}
\State{\textbf{send} $vp_r$ to $L_{v_r+1}$}\Comment{Send vote to the next leader}\label{send-vote}
\EndIf
\EndEvent

\Procedure{HighVoteReq}{$NV_{set}$}
\State{$B_{high}\gets \bot$}
\For{$s \in S$}\Comment{Iterate through all \textsc{Vote-req}s in the view change}
\State{\textbf{parse} $s$ as \msg{New-view}{v_r,vr_r,vp_r}}
\State{\textbf{parse} $vr_r$ as \msg{Vote-req}{B}}
\If{$\Call{BlockRank}{B,B_{high}}$}
\State{$B_{high}\gets B$}\Comment{Update the highest ranked block}
\EndIf
\EndFor
\State{\Return $B_{high}$}\Comment{Return highest ranked block in the view change}
\EndProcedure
\end{algorithmic}

\caption{Slow View Change}\label{alg:siesta-vc}
\end{algorithm}


\textbf{Commit Rule.} 
The commit rule determines which blocks in the chain can be safely marked as committed; it is invoked each time a replica receives a valid \msg{Vote-req}{B'} message from the leader. The test considers the last two QCs and their associated blocks (Alg. \ref{alg:commit} lines \ref{gather-begin}-\ref{gather-end}). Informally, a block is safe to commit when no possible conflicting block can also be committed, in other words when no conflicting QC could have formed. More specifically, the commit test considers two cases. We write $QC_{child}$ and $QC_{parent}$ to denote the respectively the last and second to last QCs in the chain.\\
\ul{\textit{Consecutive QCs.}} If the blocks certified by $QC_{parent}$ ($B_{parent}$) and $QC_{child}$ ($B_{child}$) were proposed in consecutive views (Alg. \ref{alg:commit} line \ref{check-consec}), it safe to commit $B_{parent}$. As shown in \S\ref{s:relax-take-it-eeeaaasy}, no higher ranked (than $B_{parent}$) conflicting QC will form.\\
\ul{\textit{Non-consecutive QCs.}} The use of \phaseone{} messages precludes conflicting QCs from forming in the presence of omission faults or asynchrony. It does not, however, prevent conflicting QCs from forming when the leader equivocates. Thus, the first step is to identify whether a conflicting QC could have formed as a result of equivocation. 
This is done by iterating through all of the ancestor blocks in between $B_{parent}$ and $B_{child}$ and looking for evidence of equivocation for a conflicting block (Alg. \ref{alg:commit} lines \ref{equiv-begin}-\ref{equiv-end}). As mentioned in \S\ref{subsec:datastructures}, equivocating blocks are different blocks proposed in the same view. Thus, evidence of equivocation (equivocation proof) consists of \textsc{Vote-req} proposal messages with equivocating blocks. It is important that this equivocation proof contains a \textsc{Vote-req} for a \textit{conflicting} block. Otherwise, this equivocation proof indicates that a non-conflicting QC could have formed, which does not violate safety.
Upon detecting equivocation, replicas must explicitly abort committing $B_{parent}$ (Alg. \ref{alg:commit} lines \ref{equiv-detect-begin}-\ref{equiv-detect-end}). 
Note that aborting in this case does not violate Theorem \ref{siesta-main}: we show in our proofs that the existence of an equivocation proof for a conflicting block guarantees that the leader who proposed $B_{parent}$ must have equivocated, and thus is Byzantine faulty.
Otherwise, if no equivocation proof is found, the replica can safely commit $B_{parent}$ (Alg. \ref{alg:commit} line \ref{safe-commit}). Our full correctness proofs are in Appendix \ref{sec:proofs}.

\begin{algorithm}
\begin{algorithmic}[1]

\Event{receiving a valid $vr_r \gets$ \msg{Vote-req}{B}}
\State{$QC_{child} \gets B.QC_{anc}$}\Comment{This is the last QC in the chain}\label{gather-begin}
\State{$B_{child} \gets QC_{child}.B$}
\State{$QC_{parent} \gets B_{child}.QC_{anc}$}\Comment{This is the second to last QC in the chain}
\State{$B_{parent} \gets QC_{parent}.B$}\label{gather-end}
\If{$\Call{AreConsecutiveQCs}{QC_{parent},QC_{child}}$} \Comment{Consecutive views, safe to commit}\label{check-consec}
\State{\textbf{commit} $B_{parent}$}
\Else
\State{$C \gets \emptyset$}\Comment{Keeps track of possible conflicting QC blocks}
\For{$B_{anc} \in \Call{GetAncestors}{B_{parent}, B_{child}}$}\label{equiv-begin}
\State{$//$ \textit{look for possible conflicting QCs by finding equivocation proofs}}
\State{$C \gets \Call{FindEquivProof}{B_{parent}, B_{anc}.p,B_{anc}.NV_{set}}$}\label{equiv-end}
\If{$C \neq \emptyset$}\label{equiv-detect-begin}
\State{\textbf{abort}}\Comment{Possible conflicting QC not safe to commit}\label{equiv-detect-end}
\EndIf
\EndFor
\State{\textbf{commit} $B_{parent}$}\label{safe-commit}
\EndIf
\EndEvent

\Procedure{AreConsecutiveQCs}{$QC_{parent},QC_{child}$}
\State{\Return $QC_{parent}.B.v+1=QC_{child}.B.v$}
\EndProcedure

\Procedure{GetAncestors}{$B_{parent}, B_{child}$}
\State{$A \gets \emptyset$}
\State{$B_{anc} \gets B_{child}$}
\While{$B_{anc}.v>B_{parent}.v$}
\State{$A \gets A \cup B_{anc}$}
\State{$B_{anc}=B_{anc}.p$}
\EndWhile
\State{\Return $A$}
\EndProcedure

\Procedure{FindEquivProof}{$B_{parent}, B_{target}, S$}
\State{$C \gets \emptyset$}\Comment{Keep track of conflicting blocks}
\For{$s\in S$}\Comment{Iterate through all blocks}
\State{\textbf{parse} $s$ as \msg{Vote-req}{B'}}
\State{$//$ \textit{Different vote-reqs with the same view indicate equivocation!}}
\If{$B'.v=B_{target}.v \land B'\neq B_{target} \land B'$ \textit{conflicts} with $B_{parent}$}
\State{$C \gets C \cup B'$}
\EndIf
\EndFor
\State{\Return $C$}\Comment{Otherwise, no equivocation}
\EndProcedure

\end{algorithmic}

\caption{Commit Rule}\label{alg:commit}
\end{algorithm}
\section{Complexity and Performance Results}
\label{sec:complexity}

\sys{} is the first CLSO protocol to satisfy AHL. In the next section, we quantify the theoretical/pratictal benefits/tradeoffs of our approach. We summarize the main properties of \sys{} as compared to the state of the art CLSO protocols in Table~\ref{comparison}. Specifically, we measure the word communication complexity of each protocol excluding the view synchronizer, where a word contains a constant amount of signatures or bits. Word complexity measures the amount of words sent by honest parties over all possible executions and adversarial strategies. 
We say a protocol is responsive if after GST the latency between consecutive honest leaders is $O(\delta)$, where $\delta$ is the actual network delay. \sys{} is the first protocol to satisfy \nonconsec{}. It does so while maintaining quadratic word complexity with threshold signatures, linear word complexity with SNARKs, optimal phase complexity, and responsiveness with consecutive honest leaders. Analysis of communication complexity and responsiveness can be found in Appendix \ref{sec:proofs}.
\neil{TODO: Add analysis to appendix}

\begin{table*}[!t]
\caption{Comparison of CLSO BFT protocols (excluding view synchronizer)
}
\centering
\def\arraystretch{1.2}
\begin{tabular}{ |c|c|c|c|c|c| } 
 \hline
 \textbf{Protocol} & \makecell{\textbf{Complexity}\\\textbf{(thresh)}} & \makecell{\textbf{ Complexity}\\\textbf{(SNARKs)}} & \makecell{\textbf{\# of}\\\textbf{phases}} & \makecell{\textbf{Responsive}\\\textbf{(consec.)}} &
 \makecell{\textbf{\nonconsec{}}} \\
 \hline
Casper FFG~\cite{buterin2017casper} & $O(n)$ & $O(n)$ & 2 & No & No \\
 HotStuff~\cite{yin2019hotstuff} & $O(n)$ & $O(n)$ & 3 & Yes & No \\
 Fast-HotStuff~\cite{jalalzai2020fasthotstuff} & $O(n^2)$ & $O(n)$ & 2 & Yes & No \\ 
Jolteon~\cite{jolteon-ditto} & $O(n^2)$ & $O(n)$ & 2 & Yes & No \\ 
 \textbf{\sys{}} & $\mathbf{O(n^2)}$ & $\mathbf{O(n)}$ & \textbf{2} & \textbf{Yes} & \textbf{Yes} \\
 \hline
\end{tabular}
\label{comparison}
\end{table*}

Next, we formally quantify the performance gains made possible by strengthening the liveness condition from \consec{} to \nonconsec{}. \sys{} will commit blocks in the presence of any $k+1$ honest leaders after GST and no longer requires $k+1$ consecutive leaders. In Figure~\ref{fig:results} we compare \sys{} to 1) two-phase CLSO protocols (DiemBFTv4~\cite{diem2021}, Fast-Hotstuff~\cite{jalalzai2020fasthotstuff}, Jolteon~\cite{jolteon-ditto}), 2) three-phase CLSO protocols
(Hotstuff~\cite{yin2019hotstuff}).
We calculate the expected number of rounds necessary to commit an operation under \nonconsec{} and \consec{} when choosing leaders at random.  We additionally simulate commit latency when electing leaders in a round-robin fashion.

\begin{theorem}
With a random leader election scheme, after GST, the expected 
number of rounds necessary to commit a block under the \consec{} is $L = \frac{(1-p^k)}{(1-p)p^k}$~\cite{trials2021drekic} where $p$ = $\frac{n-f}{n}$
and $k$ is the number of consecutive honest leaders needed. 
\end{theorem}
\begin{theorem}
With a random leader election scheme, after GST with only omission faults, the expected 
number of rounds necessary to commit a block in \sys{} is $\frac{3n}{n-f}$
\end{theorem}

\begin{proof}
Recall that two-round CLSO protocols require a sequence of three honest leaders to commit an operation. Let $p$ be the probability of selecting an honest leader $p=\frac{n-f}{n}$. As leaders are independent, the number of rounds until selecting the first, second, and third honest leaders can be viewed as three independent random variables $X_1$, $X_2$, and $X_3$ with the same distribution. 
The expected number of rounds until the selecting the nth honest leader $\mathbb{E}(X_n)$ follows a geometric distribution; by definition $\mathbb{E}(X_n)=\frac{1}{p}$.
\nc{I think some of the detail got lost here? Do we want to add a citation for it, or write the proof?}
For three-rounds, we have $L=\mathbb{E}(X_1) + \mathbb{E}(X_2) + \mathbb{E}(X_3)=\frac{3}{p}$.\fs{maybe add a half sentence why? -> pick first leader, then start same distribution to find second leader, then third.. since its sequential you just add up}

\end{proof}

Next, we simulate a scenario in which leaders are elected round-robin; we mark an operation has committed when there is sufficiently many honest leaders to satisfy the protocol's commit rule.  In CLSO protocols, the number of rounds directly influences both latency and throughput. If a round has latency $x$, then commit latency for an operation will be  $x * rounds$ while throughput is calculated by dividing the batch size by the expected commit latency. We write \consec(4) for Hotstuff (requires four consecutive leaders), \consec(3) for Fast-Hotstuff, Jolteon and DiemBFTv4, and finally AHL for our own protocol \sys{}. Figure~\ref{fig:results} shows the resulting commit latency CDF. As expected, \sys{} achieves an expected commit latency of $4.5$; \consec(3) requires $\approx 7$ rounds. \consec(4) has worst expected performance, taking $12$ rounds to commit.  Worst-case observed commit latency is especially interesting: \sys{} has relatively low worst-case latency, with 18 rounds, while \consec(3) protocols have a worst-case commit time of 76. \consec(4) has seven times worst latency, with a worst-case commit time of 129 rounds.

%
%
%
%

\begin{figure}[t]
    \centering
    \includegraphics[scale=0.4]{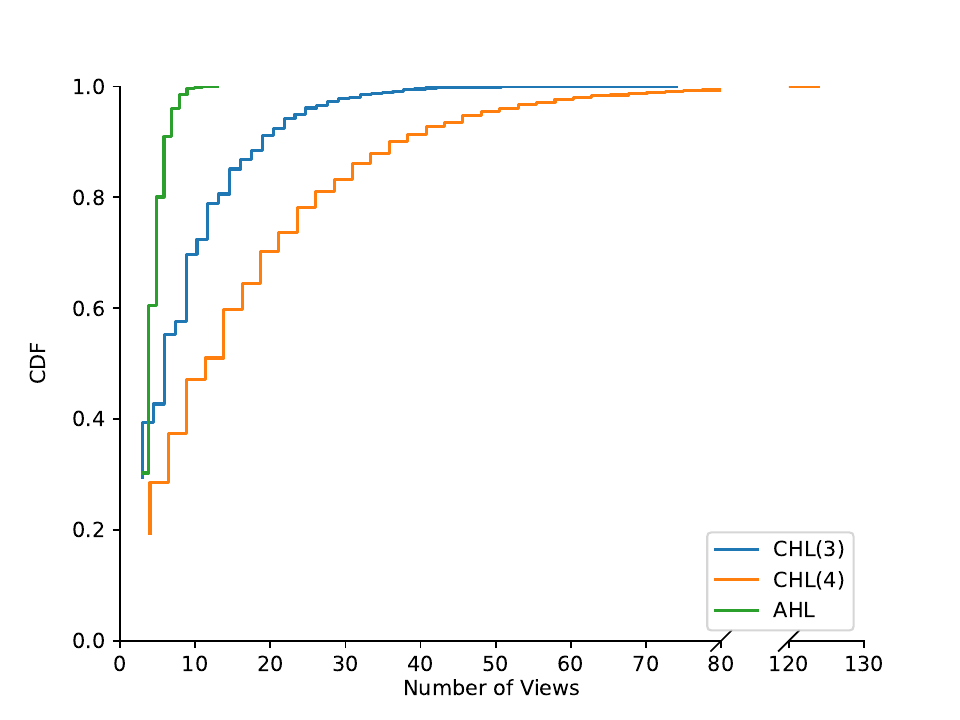}
    \caption{CDF of the number of views needed to commit an operation n = 100. \vspace{-1mm}}
    \label{fig:results}
\end{figure}
\section{Necessary Condition}
\label{sec:sequentiality}

\sys{} is the first CLSO protocol to guarantee the stronger \nonconsec{} condition while also maintaining safety. It is the only algorithm to guarantee the partial synchrony equivalent of \textit{sequentiality}~\cite{rotating-leader}. We find that this property is \textit{necessary} to support AHL, which explains why no other CLSO protocol successfully offered \nonconsec{}. 
Intuitively, sequentially states that, after GST, leaders must extend any blocks proposed by honest leaders. This is precisely what \sys{} aims to achieve via the use of \phaseone{} messages and implicit QC materialization.

\begin{definition} \textit{(Sequentiality).}
Let $L_i$ and $L_j$ be a pair of honest leaders, and wlog $i<j$. After GST, if $L_i$ sends a proposal $B_i$, then $L_j's$ proposal $B_j$ must extend $B_i$.

\end{definition}

We prove in Appendix \ref{lem:seq}:
\begin{theorem}\label{thm:siesta-seq}
    \sys{} satisfies \textbf{sequentiality}.
\end{theorem}


\medskip In contrast, other CLSO protocols only include QCs in their slow view change. A Byzantine leader can intentionally fail to form a QC for an honest block, thus precluding this block from appearing in future view changes. We prove in Appendix \ref{thm:prior-clso}:

\begin{theorem}
    Prior partially synchronous CLSO protocols \textit{do not} satisfy \textbf{sequentiality}.
\end{theorem}

\medskip Next, we show that sequentiality is, in fact, a necessary property to offering \nonconsec{}. 
We prove in Appendix \ref{thm:appendix-seq}: 

\begin{theorem}\label{thm:seq}
    \nonconsec{} CLSO is achievable only if \textbf{sequentiality} is satisfied.
\end{theorem}

\section{Conclusion}
\label{sec:conclusion}
This paper introduces \sys{}, the first CLSO protocol to guarantee that, after GST, the proposal of an honest leader will be committed after two honest views. In contrast, all other  CLSO protocols require three (or four) consecutive honest leaders to commit a block. \sys~observes that, to offer \nonconsec, a protocol must guarantee sequentiality, and that sequentiality can only be enforced through careful use of \phaseone~messages. These are messages that are instead traditionally discarded during view changes by prior work. \sys's stronger liveness guarantee allows it to outperform other CLSO protocols by up to 4x.


\bibliographystyle{ACM-Reference-Format}
\bibliography{references}

\clearpage
\appendix
\section{Proofs}\label{sec:proofs}

\subsection{Sequentiality}
\begin{theorem}\label{thm:prior-clso}
    Prior partially synchronous CLSO protocols \textit{do not} satisfy \textbf{sequentiality}.
\end{theorem}

\begin{proof}
We prove by counterexample \nc{for which protocol}, showing an execution which violates sequentiality. 
Let $GST=0$. In view $1$, $L_1$ is honest and multicasts its proposal $B_1$ to all replicas. All honest
replicas send votes for $B_1$ to $L_2$. 
However, in view $2$, $L_2$ is faulty and crashes (this can be before it receives messages, or even after it forms a QC 
for $B_1$). All of the
honest replicas timeout and in view $3$, $L_3$ is honest and collects timeout
messages from $n-f$ replicas excluding $L_2$. Note, that each timeout message contains \textit{only} the highest QC observed.
The highest QC that $p_3$ learns from the timeout messages is the genesis QC, so $L_3$'s proposal extends genesis. This 
violates sequentiality because $GST=0$ and $L_1$ is honest, and thus $L_3$ must extend $B_1$.    
\end{proof}

\begin{theorem}\label{thm:appendix-seq}
    \nonconsec{} CLSO is achievable only if \textbf{sequentiality} is satisfied.
\end{theorem}

\begin{proof}
Suppose for the sake of contradiction there exists a \nonconsec{} CLSO protocol,
$\Pi$, which does not satisfy sequentiality. Since $\Pi$ does
not satisfy sequentiality there must exist views $i,j$ where $i<j$ after GST,
such that $L_i$ (leader of view $i$) and $L_j$ (leader of view $j$) are honest, and $L_j$'s proposal does not extend $L_i$'s proposal.
We show an execution of $\Pi$ that violates the safety property of 
BFT-SMR. Let $GST=0$, $L_i$ and $L_j$ be honest, and all 
leaders in between views $i$ and $j$ be Byzantine. $L_i$ and $L_j$ broadcast proposals $B_i$ and $B_j$ at heights $h_i$ and $h_j$ respectively. 
Since honest leaders are elected infinitely often (CLSO property), there must
exist honest leaders $L_{i_1}, L_{i_2},...,L_{i_k}$, in views $i_1,...,i_k>i$; and likewise honest leaders $L_{j_1},L_{j_2},...,L_{j_k}$, in views $j_1,...,j_k>j$.
Since $\Pi$ satisfies \nonconsec{}, all honest 
replicas will commit both $B_i$, $B_j$ and all of its ancestors after views $i_k$ and $j_k$ respectively. 
Since, by assumption, $j>i$ and $B_j$ does not extend $B_i$, there must exist an ancestor proposal $B_l\neq B_i$ at height $h_i$ of the log from which $B_j$ descends.
This violates safety for BFT-SMR since honest replicas committed different proposals at the same height.
\end{proof}

\subsection{Safety}

\begin{definition}
(Safety). Honest replicas do not commit conflicting blocks.
\end{definition}

\begin{definition}
(Block Extension). A block $B'$ extends block $B$ (denoted as $B \longleftarrow B'$), if $B$
is an ancestor for $B'$: there must exist a path of parent blocks from $B'$ to $B$.
\end{definition}

\begin{definition}
(QC Extension). A block $B'$ extends certificate $QC_{B}$ (denoted as $QC_{B} \longleftarrow B'$), if $B'.QC_{anc}$ certifies a block, $B^*$ such that $B \longleftarrow B^*$
\end{definition}

\begin{definition}
(Conflicting blocks). Two blocks, $B$ and $B'$ conflict, if 
$\neg (B \longleftarrow B' \lor B' \longleftarrow B)$ ($B$ and $B'$ do not extend each other).
\end{definition}

\begin{definition}
(Equivocation Proof). An equivocation proof in view $v$, $\pi_v$, consists
of at least one pair of \msg{Vote-Req}{B} and \msg{Vote-Req}{B'}, where
$B.v=B'.v$ and $B \neq B'$.
\end{definition}

\begin{definition}
(Equivocation extension). A block $B'$ extends an equivocation in view $v$ 
(denoted as $\pi_v \longleftarrow B'$), if there exists a block $B^*$ such
that $B^* \longleftarrow B'$ and $B^*.S$ contains \msg{Vote-Req}{B_c}, where
$B_p \longleftarrow B_c$, and $B_p$ contains an equivocation proof for view 
$v$.
\end{definition}

\begin{lemma}\label{safety1}
For any two certified blocks, $B$ and $B'$ where $B.v=B'.v$, then $B=B'$.
\end{lemma}

\begin{proof}
Suppose for the sake of contradiction that $B \neq B'$.
This means that $n-f$ replicas voted for both $B$ and $B'$ in the same view 
($B.v=B'.v$). By quorum intersection,
at least one honest replica voted for both $B$ and $B'$, a contradiction.
\end{proof}

\begin{lemma}\label{safety2}
If there exists a certified block $B$, and a valid \msg{Vote-req}{B'} message, where $B'.v>B.v$ and
$B \nleftarrow B'$, then $\pi_{B.v} \longleftarrow B'$, where $\pi_{B.v}$
is an equivocation proof containing \msg{Vote-Req}{B}.
\end{lemma}

\begin{proof}
Suppose for the sake of contradiction that $B'$ does not extend an equivocation
proof $\pi_{B.v}$ containing \msg{Vote-Req}{B}. Let $B_f$ be the earliest (lowest view) ancestor block of $B'$, where $B_f.v>B.v$ and $B \nleftarrow 
B_f$. Also let
$B_p$ be $B_f$'s direct parent block. By the definition of $B_f$, $B_p$ cannot
conflict with $B$. Therefore, either
$B \longleftarrow B_p$ OR $B_p \longleftarrow B$ or $B_p.v=B.v$. 
We now consider each of the three cases individually.

\textbf{Case 1: $B \longleftarrow B_p$}. Since $B_p \longleftarrow B_f$ and
$B \longleftarrow B_p$, then $B \longleftarrow B_f$, a contradiction since
$B \nleftarrow B_f$.

\textbf{Case 2: $B_p \longleftarrow B$}. This implies $B_p.v<B.v$ since blocks
can only extend blocks in lower views. Since $B$ was certified, a quorum of
$n-f$ replicas updated their $vr_r$ to be \msg{Vote-Req}{B}. By quorum
intersection, at least one \textsc{Vote-Req} in $B_f.v$ must contain
\msg{Vote-Req}{B}. This is a contradiction since $B_f$'s parent is $B_p$,
which is in a lower view than $B$, even though $B_f.S$ contains
\msg{Vote-Req}{B}.

\textbf{Case 3: $B_p.v=B.v$}. If $B_p=B$, then $B \longleftarrow B_f \longleftarrow B'$, a contradiction. Otherwise $B_p \neq B$. Since $B$ was
certified by quorum intersection, $B_f.S$ must contain \msg{Vote-Req}{B}. And
since $B_f$'s direct parent is $B_p$ then $B_f.S$ must also contain 
\msg{Vote-Req}{B_p}. This constitutes an equivocation proof since $B.v=B_p.v$
and $B \neq B_p$. Since $B_f$ extends an equivocation
proof and $B_f \longleftarrow B'$, $\pi_{B.v} \longleftarrow B'$,
a contradiction.
\end{proof}

\begin{lemma}\label{safety3}
If an honest replica commits block $B$ (let $QC_B$ be the QC that certifies 
$B$, $QC_c$ be the QC that causes the replica to commit $B$, and $B_c$ the 
block that is certified by $QC_c$), then for every certified block $B'$ where 
$B.v<B'.v<B_c.v$, then $B \longleftarrow B'$.
\end{lemma}

\begin{proof}
Suppose for the sake of contradiction there exists a certified block $B'$ such
that $B.v<B'.v<B_c.v$ and $B \nleftarrow B'$. Since $B'$ is certified there must exist a valid
\msg{Vote-req}{B'} message. By lemma~\ref{safety2} $B'$ must have an ancestor block,
which contains an equivocation proof containing \msg{Vote-Req}{B} for view $B.v$. 
Since $B \longleftarrow B_c$
and $B \nleftarrow B'$, $B' \nleftarrow B_c$. Again applying 
lemma~\ref{safety2}, $B_c$ must have an ancestor block, which contains an
equivocation proof for $B'.v$ (with a \msg{Vote-Req}{B'} message). This is a contradiction, since an honest replica
committed $B$, which means it checked that every ancestor block of $B_c$ did not
contain a conflicting \textsc{Vote-req}; 
however, an ancestor block of $B_c$
contains a \msg{Vote-Req}{B'} message, and $B'$ has an ancestor block which 
contains an equivocation proof for $B.v$.
\end{proof}

\begin{lemma}\label{safety4}
If an honest replica commits block $B$ after receiving certified block $B_c$,
then for every valid \msg{Vote-Req}{B'} such that $B'.v>B_c.v$, 
$B \longleftarrow B'$.
\end{lemma}

\begin{proof}
We now prove the lemma by induction on view numbers $v'$ such that $v'=B'.v$,
and $v'>B_c.v$.

\textbf{Base case: } Let $v'=B_c.v+1$. We now consider \msg{Vote-Req}{B'}, where
$B'.v=v'$.
By lemmas~\ref{safety1} and~\ref{safety2} any certified block with
a view $\geq B.v$ must extend $B$. We now consider $B_c$'s direct parent block,
$B_p$. Since $B_c$ was certified, by quorum
intersection there must be at least one \msg{Vote-Req}{B_c}, where $B_c.QC_{anc}.B=B$, 
in $B'.NV_{set}$. Since $B'$ is valid, $B_p.v=B_c.v$. If $B_p.QC \neq B_c.QC_{anc}$, then
either $B_p.QC_{anc}.B.v<B_c.QC_{anc}.B.v$ or $B_p.QC_{anc}.B.v>B_c.QC_{anc}.B.v$.
If $B_p.QC_{anc}.B.v<B_c.QC_{anc}.B.v$, then $B_c$ has a higher QC, so $B'$
must extend $B_c$. Otherwise, $B_c.QC_{anc}.B.v=B.v$, $B_p.QC_{anc}.B.v>B.v$. By
lemma~\ref{safety2}, any quorum certificate with view between $B.v$ and $B_c.v$
must extend $B.v$, therefore $B \longleftarrow B_p$, and so $B \longleftarrow B'$.

\textbf{Induction Step: } We assume the lemma holds for all $v'-1>B.v$, and now we
consider $v'$. We now consider \msg{Vote-Req}{B'}, where $B'.v=v'$, and the direct
parent of $B'$ is $B_p$. By the base case and induction assumption if 
$B_c.v<B_p.v<v'$, then $B \longleftarrow B_p$, and so $B \longleftarrow B'$.
Otherwise, $B_p.v \leq B_c.v$. Since $B_c$ was certified by quorum intersection,
there must be at least one \msg{Vote-Req}{B_c}, where $B_c.QC_{anc}.B=B$, 
in $B'.NV_{set}$. Since $B'$ is valid, $B_p.v=B_c.v$. If $B_p.QC_{anc} \neq B_c.QC_{anc}$, then
either $B_p.QC_{anc}.B.v<B_c.QC_{anc}.B.v$ or $B_p.QC_{anc}.B.v>B_c.QC_{anc}.B.v$.
If $B_p.QC_{anc}.B.v<B_c.QC_{anc}.B.v$, then $B_c$ has a higher QC, so $B'$
must extend $B_c$. Otherwise, $B_c.QC_{anc}.B.v=B.v$, $B_p.QC_{anc}.B.v>B.v$. By
lemma~\ref{safety2}, any quorum certificate with view between $B.v$ and $B_c.v$
must extend $B.v$, therefore $B \longleftarrow B_p$, and so $B \longleftarrow B'$.
\end{proof}

\begin{lemma}\label{safety5}
If an honest replica commits block $B$, then for every certified block $B'$ such
that $B'.v>B.v$, $B \longleftarrow B'$.
\end{lemma}

\begin{proof}
By lemmas~\ref{safety1}, ~\ref{safety2}, and~\ref{safety3}, 
any certified block $B'$ such that
$B.v<B'.v \leq B_c.v$ must extend $B$. By lemma~\ref{safety4}, any valid
\msg{Vote-Req}{B'} where $B'.v>B_c.v$, must extend $B$. Therefore any
certified block with view $> B_c.v$ must also extend $B$.
\end{proof}

\begin{lemma}\label{safety6}
For any two blocks, $B$ and $B'$ committed by honest replicas, either
$B \longleftarrow B'$ or $B' \longleftarrow B$.
\end{lemma}

\begin{proof}
Note that any committed block must be certified. By lemma~\ref{safety1}, 
$B \neq B'$. If $B.v>B'.v$ then by lemma~\ref{safety5}, $B' \longleftarrow B$.
If $B.v<B'.v$ then by lemma~\ref{safety5} $B \longleftarrow B'$.
\end{proof}

\subsection{Liveness}
Like all other partially synchronous BFT SMR protocols, liveness depends on a 
critical component called the view synchronizer. 
There are many synchronizers such as Cogsworth~\cite{naor2020cogsworth} 
and FastSync~\cite{bravo2020making}. Any of these synchronizers can be used with 
this protocol. We assume the synchronizer satisfies the following theorem.

\begin{theorem}\label{synchronizer}
Let $v$ be a view with an honest leader after GST. Within time bound of $T_f$
from when the first honest replica enters $v$, all honest replicas also enter $v$
and received a proposed block $B_v$ from the leader of $v$.
\end{theorem}

\begin{lemma}\label{lem:seq}
If an honest leader proposes a block $B$ such that $B.v$ is a view after GST, then
for every valid \msg{Vote-Req}{B'} such that $B'.v \geq B.v$ must extend $B.v$.
\end{lemma}

\begin{proof}
We prove the lemma by induction over view numbers $v' \geq B.v$, where $v'=B.v$.

\textbf{Base Case: $v'=B.v$}. Since the leader of $B.v$ is honest, it will not
equivocate, and therefore the only valid \textsc{Vote-Req} in view $B.v$ is
for $B$. Since $B \longleftarrow B$, the base case is satisfied.

\textbf{Induction Step:}. We assume the lemma holds for all $v'-1 \geq B.v$, and
now consider view $v'$. Let $B_p$ be the direct parent block of $B'$. If 
$B_p.v \geq B.v$ then by the base case and induction assumption 
$B \longleftarrow B_p$, so $B \longleftarrow B'$, and we are done. Otherwise,
$B_p.v<B.v$. Since view $B.v$ had an honest leader and it was a view after GST,
by theorem~\ref{synchronizer} and the well-formedness of an honest leader's 
proposal all honest replicas must have entered $v$ before
their timeout expired, and updated their $vr_r$ and $vp_r$ to be for block $B$.
Therefore by quorum intersection, $B'.NV_{set}$ at the minimum is guaranteed to have
a \msg{Vote-Req}{B} message. And since $B_p.v<B.v$, $B'$ cannot directly extend
$B_p$ since there is a higher \textsc{Vote-Req} in $B'.NV_{set}$.
\end{proof}

\begin{lemma}
Let $v_s$ be a view after GST. Every honest replica eventually commits some block
$B$ with $B.v \geq v_s$.
\end{lemma}

\begin{proof}
Since the number of Byzantine replicas is bounded by $f$, we can find views
$v<v'<v''$ all with honest leaders, such that $v \geq v_s$. From 
theorem~\ref{synchronizer}, all honest replicas receive the proposed block $B_v$
from the leader of $v$ within $T_f$ time of the first honest replica entering
$v$. When instantiated with a view synchronizer such as ~\cite{naor2020cogsworth} or ~\cite{jolteon-ditto}, $T_f=2\Delta + \Delta=3\Delta$ (the extra $\Delta$) is for the 
materialization timer).
By the well-formedness of an honest leader's proposal, all honest replicas
will accept it, update their local $vr_r$ and $vp_r$, to the corresponding 
\msg{Vote-Resp}{B_v} message, and send it to the next leader. If $v'=v+1$,
then the leader of $v'$ will receive the \textsc{Vote-resp} messages within $\Delta$ time after view synchronization, form a QC for $B$ using the votes from $2f+1$ honest
replicas, and send the QC which will arrive to all honest replicas by time $T_f+2\Delta$. Otherwise, if $v'>v+1$, then the leader of $v'$ is honest and will 
wait the view timeout for the previous view. 
After view synchronization and $\Delta$ time it will receive \textsc{NewView} messages from 
all honest replicas. By lemma~\ref{lem:seq} and the fact that honest 
replicas only update their $vr_r$ for higher views, the leader of $v'$ will
receive $2f+1$ \textsc{Vote-Resp} messages for blocks that extend $B$. $L_{v'}$
will then be able to form a quorum certificate for $B$ (if it hasn't already
formed). Then by theorem~\ref{synchronizer}, $L_{v'}$'s proposed block $B'$,
which contains a quorum certificate for $B$ (or a block that extends $B$), will
be received by all honest replicas within the first honest replica entering
$v'$. By the well-formedness of $L_{v'}$'s proposal, all honest replicas will
accept it, update their local $vr_r$ and $vp_r$ to the corresponding $B'$, and
sent their votes to the next leader. If $v''=v'+1$,
then the leader of $v'$ will form a QC for $B'$ using the votes from $2f+1$ honest
replicas. Otherwise, if $v''>v'+1$, then the leader of $v''$ is honest and will 
wait the timeout. After view synchronization and $\Delta$ time it will receive \textsc{NewView} 
messages from all honest replicas. By lemma~\ref{lem:seq} and the fact that honest 
replicas only update their $vr_r$ for higher views, the leader of $v''$ will
receive $2f+1$ \textsc{Vote-Resp} messages for blocks that extend $B'$. $L_{v''}$
will then be able to form a quorum certificate for $B'$ (if it hasn't already
formed). At this point $L_{v''}$ will commit $B$ since $B$ and $B'$ have been
certified; and $L_v$ is honest, which means that it is impossible for an 
equivocation proof for view $v$ to be contained within any block. $L_{v''}$, then
sends its proposed block $B''$, and by theorem~\ref{synchronizer} all honest
replicas will receive this proposal, accept it, and then also commit $B$ since it
satisfies the commit rule.

Since it is assumed that every client transaction will be repeatedly proposed
until it is committed, then eventually every client transaction will be committed
by all honest replicas.
\end{proof}

\subsection{Communication Complexity}
Protocol logic is independent of the specific signature scheme chosen; protocol complexity, 
however, is tightly coupled with this choice. In the fast view change, the leader sends a block
containing the view number, the parent block pointer, and a QC. The view number and parent block 
pointer consist of a constant number of bits. The QC contains $n-f$ signed 
\textsc{Vote-resp} messages for a given block. These $n-f$ signatures can be compressed into $O(1)$
words by using threshold signatures~\cite{gueta2019sbft}. Therefore, the leader sends $O(1)$ words
to $O(n)$ replicas. Each replica then sends a single \textsc{Vote-resp} message to the next 
leader, resulting in a total word complexity of $O(n)$ with threshold signatures.

\par For the slow view change protocol, the leader sends a block
containing the view number, the parent block pointer, the QC, and the set of $O(n)$ 
\textsc{New-view} 
messages to $O(n)$ replicas. The view number, parent block pointer, and the QC are all $O(1)$ 
words with threshold signatures. Each \textsc{New-view} message contains a view number, 
\textsc{Vote-resp} message, and a \textsc{Vote-req} message. Each \textsc{Vote-req} message 
contains a block, which contains a view number, parent block pointer, QC, and set of 
\textsc{New-view} messages. These set of \textsc{New-view} messages are used to only ensure 
validity of the block and are retrieved as part of getting the full history of the block. 
Therefore, the leader disseminates a block containing $O(n)$ words to $O(n)$ replicas, resulting 
in a total word complexity of $O(n^2)$ with threshold signatures. As a result, 
\sys{} has an overall word 
complexity of $O(n^2)$ with threshold signatures. This complexity can be further reduced to $O(n)$
by using SNARKs to prove block validity rather than including the full set of \textsc{New-view}
messages that the leader received~\cite{abspoel2020malicious}.

\subsection{Responsiveness}
There are several definitions of responsiveness but we use the definition from~\cite{rotating-leader} adapted for partial synchrony, restated here for convenience. 

\begin{definition} (\textit{Responsiveness (Consecutive Honest)})
    We say a protocol is responsive (consecutive honest) if after GST, for any two honest leaders, $L_i$, $L_j$ for views $i<j$ (wlog) respectively, $L_j$ sends its proposal within $O(\delta)$ 
    time after view $i$ finishes.
\end{definition}

After GST, a slow view change only occurs in \sys{} if there is a faulty leader preceding an honest leader. Therefore, between consecutive honest leaders only a fast view change occurs. In the fast view change there is no timeout that is waited, therefore the next leader starts sending its proposal within $O(\delta)$ time of the previous view finishing, where $\delta$ is the actual network delay.

\end{document}